\documentclass[journal,compsoc]{IEEEtran}
\usepackage{setspace} 
\usepackage{tabularx}
\usepackage{balance}  

\def \figPath {./pic_corr}

\usepackage{cite}
\usepackage{graphicx}
\usepackage{amsbsy}
\usepackage{amssymb}
\usepackage{amsfonts}
\usepackage[procnumbered,linesnumbered,ruled,vlined]{algorithm2e}
\usepackage{bm}
\usepackage{array}
\usepackage{tabularx}
\usepackage[hyphens]{url}
\usepackage{amsmath}
\usepackage{environ}
\usepackage{mathtools}

\usepackage{outlines}
\usepackage{soul}
\setul{0.2ex}{0.1ex}
\usepackage[T1]{fontenc}
\usepackage{cases}
\usepackage{diagbox}

\usepackage{amsthm}
\newtheorem{lemma}{Lemma}
\newtheorem{theorem}{Theorem}
\newtheorem{definition}{Definition}

\newtheorem{example}{Example}
\newtheorem{corollary}{Corollary}

\newtheorem{remark}{Remark}

\DeclareMathOperator*{\argmax}{arg\,max}

\newcommand{\eqdef}{\overset{\mathrm{def}}{=\joinrel=}}

\makeatletter
\newcommand*{\indep}{%
        \mathbin{%
                \mathpalette{\@indep}{}%
        }%
}
\newcommand*{\nindep}{%
        \mathbin{
                \mathpalette{\@indep}{\not}
        }%
}
\newcommand*{\@indep}[2]{%
        \sbox0{$#1\perp\m@th$}
        \sbox2{$#1=$}
        \sbox4{$#1\vcenter{}$}
        \rlap{\copy0}
        \dimen@=\dimexpr\ht2-\ht4-.2pt\relax
        \kern\dimen@
        {#2}%
        \kern\dimen@
        \copy0 
}

\newenvironment{alignSmall}{\nobreak\small\noindent\align}{\endalign}
\newenvironment{alignFootnotesize}{\nobreak\footnotesize\noindent\align}{\endalign}
\newenvironment{alignScriptsize}{\nobreak\scriptsize\noindent\align}{\endalign}

\NewEnviron{myAlignS}{%
        \begin{singlespace}
                \vspace{-3pt}
                \begin{alignSmall}
                        \BODY
                \end{alignSmall}
                \vspace{-5pt}
        \end{singlespace}
}

\NewEnviron{myAlignSS}{%
        \begin{singlespace}
                \vspace{-3pt}
                \begin{alignFootnotesize}
                        \BODY
                \end{alignFootnotesize}
                \vspace{-5pt}
        \end{singlespace}
}

\NewEnviron{myAlignSSS}{%
        \begin{singlespace}
                \vspace{-5pt}
                \begin{alignScriptsize}
                        \BODY
                \end{alignScriptsize}
                \vspace{-5pt}
        \end{singlespace}
}

\usepackage{color,soul}
\newcommand{\cyrev}[1]{\textcolor{black}{#1}}

\begin{document}

%
\title{Quantifying Differential Privacy  in Continuous Data Release under Temporal Correlations}
%
%
%
%

\author{Yang~Cao,
        Masatoshi~Yoshikawa,
       Yonghui~Xiao, 
        and Li~Xiong 
\IEEEcompsocitemizethanks{
        \IEEEcompsocthanksitem Yang Cao, Xiong Li are with the Department
of Math and Computer Science, Emory University, Atlanta, GA, 30322.\protect\\
E-mail: \{ycao31, lxiong\}@emory.edu.
\IEEEcompsocthanksitem Masatoshi Yoshikawa is with the Department
of Social Informatics, Kyoto University, Kyoto, Japan, 606-8501.\hfil\break 
E-mail: yoshikawa@i.kyoto-u.ac.jp.
\IEEEcompsocthanksitem Yonghui Xiao is with Google Inc., Mountain View, CA, 94043\hfil\break 
E-mail: yohuxiao@gmail.com.
}
\thanks{This paper is a long version of \cite{cao_quantifying_2017}.}
}

%
%

\markboth{Journal of \LaTeX\ Class Files,~Vol.~XX, No.~X, August~201X}%
{Shell \MakeLowercase{\textit{et al.}}: Bare Demo of IEEEtran.cls for Computer Society Journals}
%



\IEEEtitleabstractindextext{%
\begin{abstract}
Differential Privacy (DP) has received increasing attention as a rigorous privacy framework.
Many existing studies employ traditional DP mechanisms (e.g., the Laplace mechanism) as primitives to continuously release private data for protecting privacy at each time point (i.e., event-level privacy), which assume that the data at different time points are independent, or that adversaries do not have knowledge of  correlation between data.
However, continuously generated data  tend to be temporally correlated, and such correlations can be acquired by adversaries.
In this paper, we investigate the potential privacy loss of  a traditional DP mechanism under temporal correlations.
First, we analyze the privacy leakage of a DP mechanism under temporal correlation that can be modeled using Markov Chain.
Our analysis reveals that, the event-level privacy loss of a DP mechanism may \textit{increase over time}.
We call the unexpected privacy loss \textit{temporal privacy leakage} (TPL).
Although TPL may increase over time, we find that its supremum may exist in some cases.
Second, we design  efficient algorithms for calculating TPL.
Third, we propose data releasing mechanisms that convert any existing DP mechanism into one against TPL.
Experiments  confirm that our approach is efficient and effective. 
\end{abstract}

}

\maketitle

\IEEEdisplaynontitleabstractindextext

%
\IEEEpeerreviewmaketitle

\IEEEraisesectionheading{\section{Introduction}\label{sec:introduction}}

%
%
%
%
\IEEEPARstart{W}{ith} the development of IoT technology, vast amount of temporal data generated by individuals are being collected, such as trajectories and web page click streams. 
The {continual publication} of statistics from these temporal data has the potential for significant social benefits such as disease surveillance
, real-time traffic monitoring
and web mining.
However, privacy concerns hinder the wider use of these data.
To this end, \textit{differential privacy under continual observation}
\cite{acs_case_2014}
\cite{bolot_private_2013}
\cite{chan_private_2011}
\cite{dwork_differential_2010-1}
\cite{dwork_differential_2010}
\cite{fan_fast:_2013}
\cite{kellaris_differentially_2014}
\cite{xiao_dpcube:_2014} 
has received increasing attention because differential priavcy provides a rigorous privacy guarantee.
Intuitively, differential privacy (DP)\cite{dwork_differential_2006} ensures that the change of any single user's data has a ``slight'' (bounded in $ \epsilon $) impact on the change in outputs.
The parameter $ \epsilon $ is defined to be a positive real number to control the  level of privacy guarantee.
Larger values of $ \epsilon $ result in larger privacy leakage.
\begin{figure}[t]
        \centering
        \includegraphics[scale=0.3]{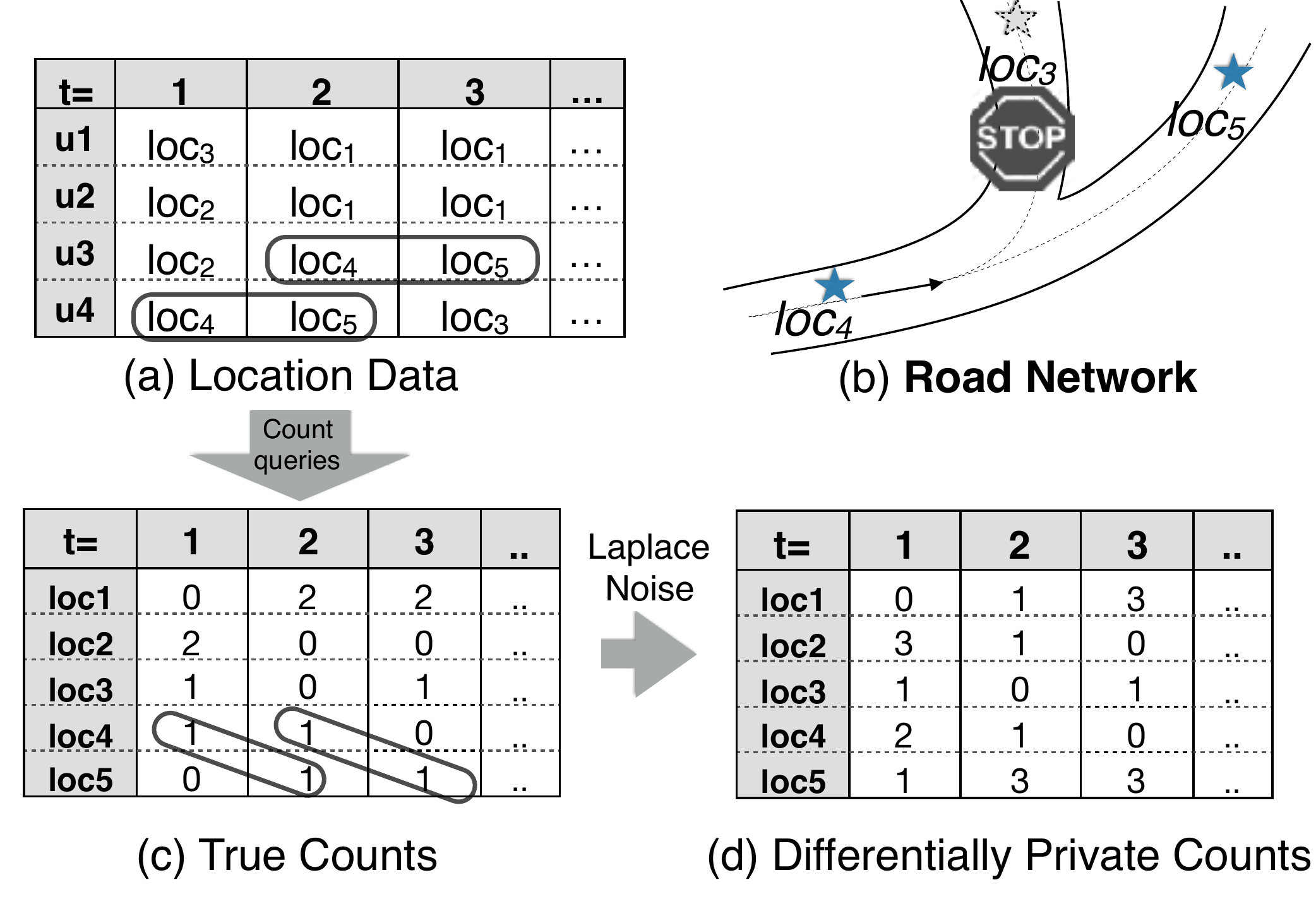} 
        \vspace{-12pt}
        \caption{Continuous Data Release with DP under Temporal Correlations.}
        \label{fig:corr_example}
        \vspace{-15pt}
\end{figure}
However, most existing works on differentially private continuous data release assume data at different time points are independent, or attackers do not have knowledge of correlation between data.
 Example \ref{example} shows that temporal correlations may degrade the expected privacy guarantee of  a DP mechanism.

\begin{example}
\label{example}
Consider the scenario of continuous data release with DP illustrated in Figure \ref{fig:corr_example}.
A trusted server collects users' locations at each time point in Figure \ref{fig:corr_example}(a) and tries to continuously publish differentially private aggregates (i.e., the private counts of people at each location).
Suppose that each user appears at only one location at each time point.
\cyrev{
According to the Laplace mechanism\cite{dwork_calibrating_2006}, adding $ Lap(2/\epsilon)$ noise\footnote{$Lap({b})$ denotes a Laplace distribution with variance $ 2b^2 $.} to perturb each count in Figure\ref{fig:corr_example}(c) can achieve $ \epsilon $-DP. 
It is because  the modification of each cell (raw data) in Figure \ref{fig:corr_example}(a) affects at most two cells (true counts) in Figure \ref{fig:corr_example}(c), i.e., the global sensitivity is $ 2 $.
However, it may not be true under the existence of temporal correlations.
For example, due to the nature of  road networks,  user  may have a particular mobility pattern such as ``always arriving at $ loc_5 $ after visiting $ loc_4 $''  (shown in Figure \ref{fig:corr_example}(b)), leading to the patterns illustrated in solid lines of Figure \ref{fig:corr_example}(a)(c). 
Such temporal correlation can be represented as {\footnotesize $ \Pr(l^{t}=loc_5|l^{t-1}=loc_4) =1$}  where $ l^t $ is the location of a user at time $ t $.
In this case,  adding $ Lap(2/\epsilon)$ noise only achieves  $2 \epsilon $-DP because the change of one cell in Figure \ref{fig:corr_example}(a) may affect four cells in Figure \ref{fig:corr_example}(c) in the worst case, i.e., the global sensitivity is $ 4 $.
}
\end{example}


Few studies in the literature investigated such potential privacy loss of event-level $ \epsilon $-DP  under temporal correlations as shown in Example \ref{example}.
A direct method (without finely utilizing the {probability} of correlation) involves amplifying the perturbation in terms of \textit{group differential privacy}\cite{chen_correlated_2014}\cite{dwork_calibrating_2006}, i.e., protecting the correlated data as a group.
In Example \ref{example}, for temporal correlation {\footnotesize $ \Pr(l^{t}=loc_5|l^{t-1}=loc_4) =1$}, we can protect the counts of $ loc_4 $ at time $ t -1$ and $ loc_5 $ at time $ t $ in a group  by increasing the scale of the perturbation (using $ Lap(4/\epsilon) $) at each time point.
However, this technique is not suitable for \textit{probabilistic} correlations to finely prevent privacy leakage and may over-perturb the data as a result.  
For example, regardless of whether {\footnotesize $ \Pr(l^{t}=loc_i|l^{t-1}=loc_i) $} is 1 or 0.1, it always protects the correlated data in a bundle.

\cyrev{
Although a few studies investigated the issue of differential privacy under {probabilistic}  correlations,  existing works are not appropriate for \textit{continuous data release} because of two reasons.
First, most of existing works focused on correlations between users (i.e., user-user correlation) rather than correlation between data at different time points (i.e., temporal correlations).
Yang et al. \cite{yang_bayesian_2015} proposed Bayesian differential privacy,  which measures the privacy leakage under probabilistic correlations between users using a Gaussian Markov Random Field without taking time factor into account.
Liu et al. \cite{changchang_liu_dependence_2016} proposed \textit{dependent differential privacy} by introducing two parameters of \textit{dependence size} and \textit{probabilistic dependence relationship} between tuples.
It is not clear whether we can specify these parameters for temporally correlated data since it is not commonly used probabilistic model.
A concurrent work \cite{song_pufferfish_2017} proposed Markov Quilt Mechanism when data correlation is represented by a Bayesian Network (BN).
Although BN is possible to model both user-user correlation (when the nodes in BN are individuals) and temporal correlation (when the nodes  in BN are  individual data at different time points), [17] assumes the private data are released only at one time.
This is the second major difference between our work and all above works: they focus on the setting of ``one-shot'' data release, which means the private data are released only at one time.
To the best of our knowledge, no study reported the dynamical change of the privacy guarantee  in continuous data release.
}

In this work, we call the adversary with knowledge of probabilistic temporal correlations \textit{adversary$ _\mathcal{T} $}.
Rigorously quantifying and bounding the privacy leakage against adversary$ _\mathcal{T} $ in continuous data release remains a challenge.
Therefore, our goal is to solve the following problems:
\begin{itemize}
        \item How do we \textit{analyze the privacy loss} of DP mechanisms against \textit{adversary$ _\mathcal{T} $}? (Section \ref{sec:privacy_analysis})
        \item How do we \textit{calculate} such privacy loss efficiently?  (Section \ref{sec:calculate_tpl})
        \item How do we \textit{bound} such privacy loss in continuous data release?  (Section \ref{sec:release_mechanism})
\end{itemize}

\subsection{Contributions}
Our contributions are summarized as follows.

First, we show that the privacy guarantee of data release \textit{at a single time} is not on its own or static; instead, due to the existence of temporal correlation, it may be increasing with previous release and even future release.
We first adopt a commonly used model Markov Chain to describe temporal correlations between data at different time points, which includes \textit{backward} and \textit{forward} correlations, i.e., $ \Pr(l_i^{t-1}|l_i^{t}) $ and $ \Pr(l_i^{t}|l_i^{t-1}) $ where $ l_i^t $ denotes the value (e.g., location) of user $ i $ at time $ t $.
We then define \textit{Temporal Privacy Leakage (TPL)}  as the privacy loss of a DP mechanism at time $ t $ against adversary$_\mathcal{T} $ who has knowledge of the above Markov model and observes the continuous data release.
We show that TPL includes two parts: \textit{Backward Privacy Leakage (BPL)} and \textit{Forward Privacy leakage (FPL)}.
Intuitively, BPL at time $ t $ is affected by previously releases that are from time $ 1 $ to $ t-1 $, and FPL at time $ t $ will be affected by future releases that are from time $ t+1 $ to the end of release.
We define $ \alpha $-\textit{differential privacy under temporal correlation}, denoted as $ \alpha$-$DP_\mathcal{T}$, to formalize the privacy guarantee of a DP mechanism against adversary$_\mathcal{T} $, which implies the temporal privacy leakage should be bounded in $ \alpha $.
We prove a new form of sequential composition theorem for {\small $ \alpha$-$DP_\mathcal{T}$}, which reveals interesting connections among event-level DP, $ w $-event DP \cite{kellaris_differentially_2014} and user-level DP{ \cite{dwork_differential_2010-1}\cite{dwork_differential_2010}}.

Second, we design efficient algorithms to calculate TPL under given backward and forward temporal correlations.
Our idea is to transform such calculation into finding an optimal solution of a \textit{Linear-Fractional Programming} problem.
This type of optimization can be solved by well-studied methods, e.g., simplex algorithm,  in exponential time.
By exploiting the constraints of this problem, we propose fast algorithms to obtain the optimal solution without directly solving the problem.
Experiments show our algorithms outperform the off-the-shelf optimization software (CPLEX) by several  orders of magnitude with the same optimal answer.

Third, we design novel data release mechanisms against  TPL effectively.
Our scheme is to carefully calibrate the privacy budget of the traditional DP mechanism at each time point to make them satisfy $ \alpha $-DP$_\mathcal{T} $.
A challenge is that TPL is dynamically changing due to previous and even future allocated privacy budgets so that $ \alpha $-DP$_\mathcal{T} $ is hard to achieve.
In our first solution, we prove that, even though TPL may increase over time,  its supremum may exist when allocating a calibrated privacy budget at each time.
That is to say, TPL will never be greater than such supremum $ \alpha $ no matter when the release ends.
However, when the releases are too short (TPL is far from reaching its supremum), 
we may over-perturb the data.
In our second solution, we design another budget allocation scheme that exactly achieves $ \alpha $-DP$_\mathcal{T} $ at each time point.

Finally, experiments confirm the efficiency and effectiveness of our TPL quantification algorithms and data release mechanisms. 
We also demonstrate the impact of different degree of temporal correlations on privacy leakage.

\section{Preliminaries}

\subsection{Differential Privacy}
Differential privacy\cite{dwork_differential_2006} is a formal definition of data privacy.
Let $ D $  be a database and $ D' $ be a copy of $ D $ that is different in any one tuple.
$ D $ and $ D' $ are \textit{neighboring databases}. 
A differentially private output from $ D $ or $ D' $  should exhibit little difference.

\begin{definition}[$ \epsilon $-DP]
        \label{def:dp}
        $\mathcal{M}$ is a randomized mechanism that takes as input $ D $ and outputs $ \bm{r} $, i.e., {\small $ \mathcal{M}(D)=\bm{r}$}.
        $\mathcal{M}$ satisfies $ \epsilon $-differential privacy if the following inequality is true for any pair of neighboring databases $ D, D'$ and all possible outputs $ \bm{r} $.
        \begin{myAlignS}
                \log\frac{\Pr(\bm{r} \in \cyrev{Range(M)}|D)}{\Pr(\bm{r} \in \cyrev{Range(M)}|D')} \leq \epsilon.  \label{eq:dp}
        \end{myAlignS}
\end{definition}

\vspace{-12pt}
The parameter {\small $ \epsilon $}, called the \textit{privacy budget}, represents the degree of privacy offered.
A commonly used method to achieve $ \epsilon $-DP is the Laplace mechanism as shown below.

\begin{theorem}[Laplace Mechanism]
        Let {\footnotesize $ Q:D\rightarrow \mathbb{R} $} be a statistical query on database $ D $.
        The sensitivity of {\footnotesize $ Q $}  is defined as the maximum $ L_1 $ norm between {\footnotesize $ Q(D) $} and {\footnotesize $ Q(D') $}, i.e., {\footnotesize $ \Delta=\max_{D,D'}||Q(D)-Q(D')||_1$}.
        We can achieve $ \epsilon $-DP by adding Laplace noise with scale {\footnotesize $ \Delta/\epsilon $}, i.e., {\footnotesize $ Lap(\Delta/\epsilon) $}.
\end{theorem}

\subsection{Privacy Leakage}
\label{subsec:pl}
\cyrev{In this section, we define the privacy leakage of a DP mechanism against  a type of adversaries $ A_i $,
who targets user $ i$'s value in the database, i.e., {\small $ l_i \in [loc_1,\ldots,loc_n]$}, and has knowledge of {\small $ D_\mathcal{K} =D-\{l_i\}$}.
}
The adversary $ A_i $ observes the private output $ \bm{r} $ and attempts to distinguish whether user $ i$'s value  is  {\small $ loc_j$}  or {\small $ loc_k$}  where {\small $  loc_j, loc_k \in [loc_1,\ldots,loc_n]$}.
%
%
We define the privacy leakage of a DP mechanism w.r.t. such adversaries as follows.
\begin{definition}
        \label{def:pl}
        The \textit{privacy leakage} of a DP mechanism $\mathcal{M}$ against one $ A_i $ with a specific $ i $ and all $ A_i, i\in\bm{U}$ are defined, respectively, as follows in which $ l_i $ and $ l_i' $ are two different possible values of  $ i $-th data in the database.
        \begin{myAlignSSS}
                &PL_0(A_i,\mathcal{M}) \eqdef \sup_{\bm{r},l_i,l_i'} \log \frac{\Pr(\bm{r}|l_i,D_\mathcal{K})}{\Pr(\bm{r}|l_i',D_\mathcal{K})} \nonumber\\
                &PL_0(\mathcal{M}) \eqdef \max_{\forall A_i,i\in\bm{U}} PL_0(A_i,\mathcal{M}) = \sup_{\bm{r},D,D'} \log  \frac{\Pr(\bm{r}|D)}{\Pr(\bm{r}|D')} \nonumber
        \end{myAlignSSS}
\end{definition}
\vspace{-10pt}
In other words, the privacy budget of a DP mechanism can be considered as a metric of \textit{privacy leakage}.
The larger {\small $ \epsilon $}, the larger the privacy leakage.
We note that a ${\epsilon}'$-DP mechanism automatically satisfies $ \epsilon $-DP if $  {\epsilon}' < {\epsilon}$. 
For convenience, in the rest of this paper, when we say that $ \mathcal{M} $ satisfies $ \epsilon $-DP, we mean that the supremum of privacy leakage is  $ \epsilon $.

\vspace{-12pt}
\subsection{Problem Setting}
\label{sec:problem_setting}
We attempt to quantify  the potential privacy loss of a DP mechanism under temporal correlations in the context of \textit{continuous data release} (e.g., releasing private counts at each time as shown in Figure \ref{fig:corr_example}).
\cyrev{For the convenience of analysis, let us assume the length of release time
is $ T $. 
Note that we do not need to know the exact $ T $ in this paper.}
Users in the database, denoted by $ \bm{U} $,  are generating data continuously.
Let {\small $ \bm{loc}=\{loc_1,\ldots,loc_n\} $} be all possible values of user's data.
We denote the value of user $ i $ at time $ t $ by $ l_i^t $. 
A trusted server collects the data of each user into the database {\small $ D^t  =\{l_1^t,\ldots,l_{|\bm{U}|}^t\} $} at each time $t$ (e.g., the columns in Figure \ref{fig:corr_example}(a)).
\cyrev{Without loss of generality, we assume that each user contributes only one data point in $ D^t $.
DP mechanisms $ \mathcal{M}^t $ release differentially private outputs $ \bm{r}^t$ independently at different time $ t $.
For simplicity, we let  $ \mathcal{M}^t $  to be the same DP mechanism  but may with different privacy budgets at each $ t \in[1,T]$.
}
Our goal is to quantify and bound the potential privacy loss of $ \mathcal{M}^t $ against adversaries with knowledge of temporal correlations. 
We note that while we use location data in Example \ref{example}, the problem setting is general for temporally correlated data.
We summarize  notations used in this paper  in Table \ref{tb:notation}.

Our problem setting is identical to \textit{differential privacy under continual observation} in the literature
\cite{acs_case_2014}
\cite{bolot_private_2013}
\cite{chan_private_2011}
\cite{dwork_differential_2010-1}
\cite{dwork_differential_2010}
\cite{fan_fast:_2013}
\cite{kellaris_differentially_2014}
\cite{xiao_dpcube:_2014}.
In contrast to ``one-shot'' data release over a static database, the attackers can observe multiple  private outputs, i.e., {\scriptsize $ \bm{r}^1,\ldots,\bm{r}^t$}.
There are typically two different privacy goals in the context of continuous data release: \textit{event-level} and \textit{user-level}
\cite{dwork_differential_2010-1}
\cite{dwork_differential_2010}. 
The former protects each user's single data point at time $ t $ (i.e., the neighboring databases are {\scriptsize $ D^t $} and {\scriptsize $ {D^t}' $}), whereas the latter protects the presence of a user with all her data on the timeline (i.e., the neighboring databases are {\scriptsize $\{ D^1 ,\ldots,D^t \}$} and {\scriptsize $\{ {D^1}',\ldots,{D^t}' \}$}).
In this work, we start from examining event-level DP under temporal correlations, and we also extend the discussion to user-level DP by studying the sequential composability of the privacy leakage.


\renewcommand{\arraystretch}{1.25}

\begin{table}[t]
	\cyrev{
        \centering
        \scriptsize
        \caption{Summary of Notations.}
        \vspace{-10pt}
        \begin{tabularx}{8.5cm}{|p{1cm}|X|}
                \hline $ \bm{U}$ & The set of users in the database \\
                \hline $ i$ & The $ i $-th user where $ i \in [1, |\bm{U}|] $\\
                \hline $ \bm{loc}$ & Value domain $ \{loc_1,\ldots, loc_n\} $ of all user's data \\ 
                \hline $ l_i^t, {l_i^t}' $&  The data of user $ i $ at time $ t $, $ l_i^t \in \bm{loc} $, $ l_i^t \neq {l_i^t}' $\\
                \hline ${D}^t$ & The database at time $ t$, ${D}^t=\{l_1^t,\ldots,l_n^t\} $\\
                \hline $ \mathcal{M}^t $ & Differentially private mechanism over $ D^t $\\
                \hline $ \bm{r}^t $ & Differentially private output  at time $t$ \\ 
                \hline $ A_i $ &   Adversary who targets user $i $ without temporal correlations \\ 
                \hline $ A_i^{\mathcal{T}} $ &  Adversary $ A_i $ with temporal correlations  \\ 
                \hline $P_i^B$ &  Transition matrix that represents {\scriptsize$ \Pr(l_i^{t-1}|l_i^t) $}, \newline i.e., backward temporal correlation, known to $ A_i^\mathcal{T} $ \\  
                \hline $P_i^F$ &  Transition matrix that represents {\scriptsize$ \Pr(l_i^{t}|l_i^{t-1}) $}, \newline i.e., forward temporal correlation, known to $ A_i^\mathcal{T} $ \\  
                \hline ${D}_\mathcal{K}^t$ & The subset of database $ D^t - \{l_i^t\} $, known to $A_i^\mathcal{T} $\\
                \hline 
        \end{tabularx}
        \label{tb:notation} 
    }
\vspace{-10pt}
\end{table}

\vspace{-1pt}
\section{Analyzing Temporal Privacy Leakage}
\label{sec:privacy_analysis}

\subsection{Adversay with Knowledge of Temporal Correlations}
\label{subsec:adv}
\vspace{-2pt}
\textbf{Markov Chain for Temporal Correlations.}
The \hyphenation{Mar-kov} Markov chain (MC) is extensively used in modeling user mobility.
For a time-homogeneous first-order MC, a user's current value only depends on the previous one.
The parameter of the MC is the \textit{transition matrix}, which describes the probabilities for transition between values.
The sum of the probabilities in each row of the transition matrix is $ 1 $.
A concrete example of transition matrix and time-reversed one  for location data  is shown in Figure \ref{fig:mc}.
As shown in Figure \ref{fig:mc}(a), if user $ i $ is at $ loc_1 $ now (time $ t $); then, the probability of coming from $ loc_3 $  (time $ t-1 $) is $0.7$, namely, {\small $ \Pr( l_i^{t-1} = loc_3 | l_i^{t} =loc_1 ) = 0.7$}.
As shown in Figure \ref{fig:mc}(b), if user $ i $ was at $ loc_3 $ at the previous time $ t-1 $, then the probability of being at $ loc_1 $ now (time $ t $) is $0.6$; namely, {\small $ \Pr( l_i^{t} = loc_1 | l_i^{t-1} =loc_3 ) = 0.6$}.
We call the transition matrices in Figure \ref{fig:mc}(a) and (b) as backward temporal correlation and forward temporal correlation, respectively.

\begin{definition}[Temporal Correlations]
        \label{def:temporal_corr}
        The backward and forward temporal correlations between user $ i $'s data  $ l_i^{t-1} $ and $  l_i^{t} $ are described by  transition matrices {\small ${P}_i^B,{P}_i^F \in \mathbb{R}^{n \times n}$}, representing {\small $ \Pr( l_i^{t-1} | l_i^{t}) $} and {\small $ \Pr( l_i^{t} | l_i^{t-1}) $}, respectively.
\end{definition}

\vspace{-1pt}
It is reasonable to consider that the backward and/or forward temporal correlations could be acquired by adversaries.
For example, the adversaries can learn them from user's historical trajectories (or the reversed trajectories) by well studied methods such as 
Maximum Likelihood estimation (supervised) or Baum-Welch algorithm (unsupervised). 
Also, if the initial distribution of {\small $ l_i^1 $}  is known (i.e., {\small $ \Pr(l_i^1) $}), the backward temporal correlation (i.e., {\small $ \Pr(l_i^{t-1}|l_i^{t}) $}) can be derived from the forward temporal correlation (i.e., {\small $ \Pr(l_i^{t}|l_i^{t-1}) $}) by Bayesian inference. 
We assume the attackers' knowledge about temporal correlations is given in our framework.

We now define  an enhanced version of $ A_i $ (in Definition \ref{def:pl}) with knowledge of temporal correlations.

\begin{definition}[Adversary$ _\mathcal{T }$]
        \label{def:adversary_t}
        Adversary$ _\mathcal{T }$ is a class of adversaries who have knowledge of (1)  all other users' data {\small $ D_\mathcal{K}^t $} at  each time $ t $ except the one of the targeted victim, i.e., {\small $ D_\mathcal{K}^t=D^t-\{l_i^t\} $}, and (2) backward and/or forward temporal correlations represented as transition matrices {\footnotesize$ P_i^B $} and {\footnotesize $ P_i^F $}.
        We denote Adversary$ _\mathcal{T }$ who targets user $ i $ by {\footnotesize $ A_i^\mathcal{T}({P}_i^B, {P}_i^F)$}.
\end{definition}

\begin{figure}[t]
        \centering
        \includegraphics[scale=0.25]{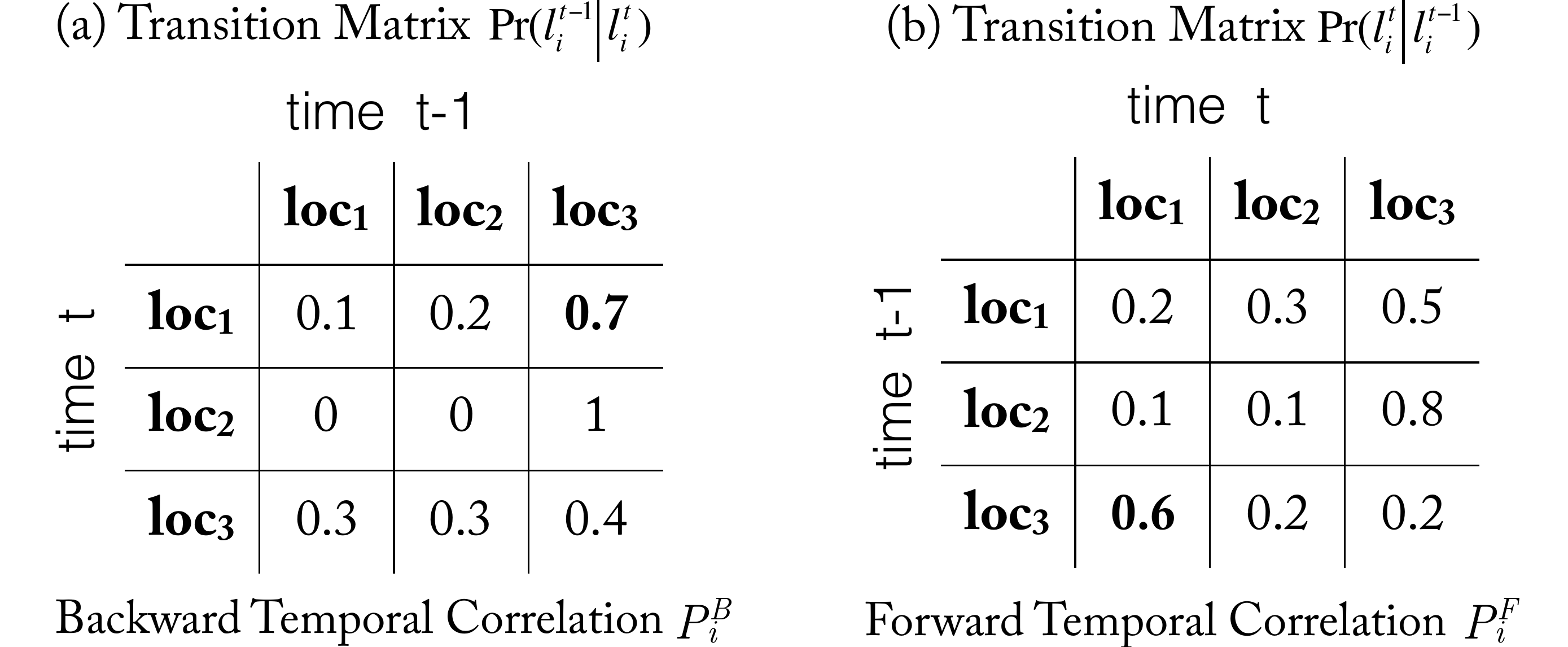} 
           \vspace{-8pt}
        \caption{Examples of Temporal Correlations.}
        \label{fig:mc}
        \vspace{-15pt}
\end{figure}

\vspace{-5pt}
There are three types of adversary$ _\mathcal{T} $: 
(i) {\footnotesize $ A_i^\mathcal{T}({P}_i^B, \emptyset) $}, (ii) {\footnotesize $ A_i^\mathcal{T}(\emptyset, {P}_i^F)$},  (iii) {\footnotesize $ A_i^\mathcal{T}({{P}_i^B,{P}_i^F})$},
where $ \emptyset $ denotes that  the corresponding correlations are not known to the adversaries.
For simplicity, we denote types (i) and (ii) as {\footnotesize $ A_i^\mathcal{T}({P}_i^B) $} and  {\footnotesize $ A_i^\mathcal{T}({P}_i^F)$}, respectively.
We note that {\footnotesize $ A_i^\mathcal{T}(\emptyset, \emptyset) $} is the same as the traditional DP adversary $ A_i $ without any knowledge of temporal correlations. 
\begin{figure*}
        \centering
        \includegraphics[scale=0.38]{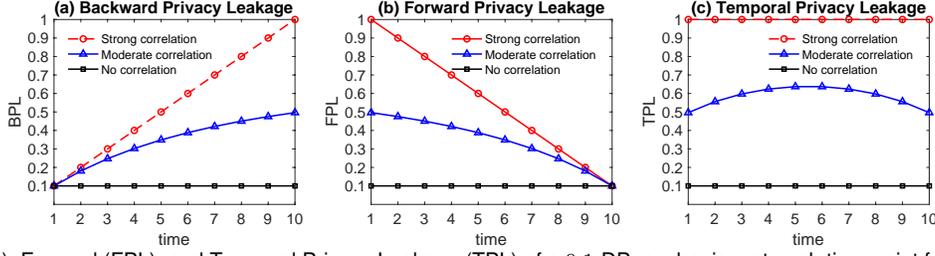} 
          \vspace{-12pt}
        \caption{Backward (BPL), Forward (FPL), and Temporal Privacy Leakage (TPL) of a $ 0.1$-DP mechanism at each time point for Example \ref{example:bpl} and  \ref{example:fpl}.}
        \label{fig:tpl}
        \vspace{-15pt}
\end{figure*}




\vspace{-1pt}
\subsection{Temporal Privacy Leakage}
\label{subsec:tpl} 
We now define the privacy leakage w.r.t. adversary$ _\mathcal{T} $.
The adversary {\small $ A_i^\mathcal{T} $} observes the differentially private outputs $ \bm{r}^t $, which is released by a traditional differentially private mechanisms $\mathcal{M}^t$ (e.g., Laplace mechanism) at each time point $ t \in [1,T] $, and attempts to infer the possible value of user $ i $'s  data at  $ t $, namely $ l_i^t $.
Similar to Definition \ref{def:pl}, we define the privacy leakage in terms of event-level differential privacy in the context of continual data release as described in Section \ref{sec:problem_setting}.

\begin{definition}[Temporal Privacy Leakage, TPL]
        \label{def:tpl}
        Let {\small $ {D^t}' $} be a neighboring database of {\small $ D^t $}.
        Let {\small $D_{\cal K}^t $} be the tuple knowledge of {\small $ A_i^\mathcal{T} $}.
        We have {\footnotesize $ {D^t}'=D_{\cal K}^t  \cup \{{l_i^t}\} $} and {\footnotesize $ {D^t}'=D_{\cal K}^t  \cup \{{l_i^t}'\} $} where $l_i^t  $ and $ {l_i^t}' $ are two different values of user $ i $'s data at time $ t $.
        Temporal Privacy Leakage (TPL)  of {\small $ \mathcal{M}^t $} w.r.t. a single {\small $ A_i^\mathcal{T} $}  and all {\small $ A_i^\mathcal{T}, i\in \bm{U}$}  are defined,  respectively, as follows. 
        \begin{myAlignSSS}
                {\textit{TPL}}&(A_i^\mathcal{T} ,\mathcal{M}^t) &\eqdef &
                \sup_{\substack{ l_i^t, {l_i^t}', \bm{r}^1, \ldots, \bm{r}^T} } \log
                {{\Pr ({\bm{r}^1}, \ldots, {\bm{r}^{T}}|{l_i^t,D_\mathcal{K}^t})} \over {\Pr ({\bm{r}^{1}}, \ldots ,{\bm{r}^{T}}|{l_i^t}',{D_\mathcal{K}^t})}}. \label{eq:tpl1} 
                \\
                {\textit{TPL}}&(\mathcal{M}^t) &\eqdef &\max_{\forall A_i^\mathcal{T}, i\in \bm{U}} {TPL}(A_i^\mathcal{T} ,\mathcal{M}^t) \label{eq:tpl2}
                \\
                &&=& \sup_{\substack{ D^t, {D^t}', \bm{r}^1, \ldots, \bm{r}^T} } \log
                {{\Pr ({\bm{r}^1}, \ldots, {\bm{r}^{T}}|{D^t})} \over {\Pr ({\bm{r}^{1}}, \ldots ,{\bm{r}^{T}}|{D^t}')}}. \label{eq:tpl3}
        \end{myAlignSSS}
\end{definition}

We first analyze {\small $ {\textit{TPL}}(A_i^\mathcal{T} ,\mathcal{M}^t) $} (i.e., Equation \eqref{eq:tpl1}) because it is key to solve Equation \eqref{eq:tpl2} or \eqref{eq:tpl3}.
\cyrev{
\begin{theorem}
	\label{thm:analysis_tpl}
	We can rewrite  {\small $ {\textit{TPL}}(A_i^\mathcal{T} ,\mathcal{M}^t) $} as follows.
	\begin{myAlignSSS}
		&
		\underbrace{
			\sup_{\substack{\bm{r}^1,...,\bm{r}^t,\\{l_i^t},{l_i^t}'}} 
			\log \frac{\Pr(\bm{r}^1,...,\bm{r}^t|{l_i^t,D_\mathcal{K}^t})}{\Pr({\bm{r}^1,...,\bm{r}^t|{{l_i^t}',D_\mathcal{K}^t}})}
		}_{\text{{\footnotesize backward privacy leakage}}}
		+
		\underbrace{
			\sup_{\substack{\bm{r}^t,...,\bm{r}^T,\\ {l_i^t},{l_i^t}'}} 
			\log \frac{\Pr(\bm{r}^t,...,\bm{r}^T|{l_i^t,D_\mathcal{K}^t})}{\Pr(\bm{r}^t,...,{\bm{r}^T|{{l_i^t}',D_\mathcal{K}^t}})}
		}_{\text{{\footnotesize forward privacy leakage}}}  \nonumber\\
		&-
		\underbrace{
			\sup_{\bm{r}^t,{l_i^t},{l_i^t}'} 
			\log \frac{\Pr(\bm{r}^t|{l_i^t,D_\mathcal{K}^t})}{\Pr({\bm{r}^t|{{l_i^t}',D_\mathcal{K}^t}})}
		}_{\textit{PL}_0(A_i^\mathcal{T},\mathcal{M}^t)}  \label{eq:tpl_expanded}
	\end{myAlignSSS}
\end{theorem}
}

It is clear that {\footnotesize $ \textit{PL}_0(A_i^\mathcal{T},\mathcal{M}^t)  = \textit{PL}_0(A_i,\mathcal{M}^t) $} because {\small $ \textit{PL}_0 $}  indicates the privacy leakage w.r.t. one output $ \bm{r} $ (refer to Definition \ref{def:pl}). 
As annotated in the above equation, we define backward and forward privacy leakage as follows.

\begin{definition}[Backward Privacy Leakage, BPL]
        The privacy leakage of {\footnotesize $ \mathcal{M}^t $} caused by {\footnotesize $ \bm{r}^1,...,\bm{r}^t $} w.r.t. {\footnotesize $ A_i^\mathcal{T} $} is called backward privacy leakage, defined as follows.
        \begin{myAlignSSS}
                {\textit{BPL}}&(A_i^\mathcal{T} ,\mathcal{M}^t) &\eqdef &
                \sup_{\substack{ l_i^t, {l_i^t}', \bm{r}^1, \ldots, \bm{r}^t} } \log
                {{\Pr ({\bm{r}^1}, \ldots, {\bm{r}^{t}}|{l_i^t,D_\mathcal{K}^t})} \over {\Pr ({\bm{r}^{1}}, \ldots ,{\bm{r}^{t}}|{l_i^t}',{D_\mathcal{K}^t})}}. \label{eq:bpl1} 
                \\
                {\textit{BPL}}&(\mathcal{M}^t) &\eqdef & \max_{\forall A_i^\mathcal{T},i \in \bm{U}} {\textit{BPL}}(A_i^\mathcal{T} ,\mathcal{M}^t) . \label{eq:bpl2} 
        \end{myAlignSSS}
\end{definition}

\begin{definition}[Forward Privacy Leakage, FPL]
        The privacy leakage of {\small $ \mathcal{M}^t $} caused by {\small $ \bm{r}^t,...,\bm{r}^T$} w.r.t. {\small $ A_i^\mathcal{T} $} is called forward privacy leakage, defined by follows.
        \vspace{-3pt}
        \begin{myAlignSSS}
                {\textit{FPL}}&(A_i^\mathcal{T} ,\mathcal{M}^t) &\eqdef &
                \sup_{\substack{ l_i^t, {l_i^t}', \bm{r}^t, \ldots, \bm{r}^T} } \log
                {{\Pr ({\bm{r}^t}, \ldots, {\bm{r}^{T}}|{l_i^t,D_\mathcal{K}^t})} \over {\Pr ({\bm{r}^{t}}, \ldots ,{\bm{r}^{T}}|{l_i^t}',{D_\mathcal{K}^t})}}. \label{eq:fpl1} 
                \\
                {\textit{FPL}}&(\mathcal{M}^t) &\eqdef & \max_{\forall A_i^\mathcal{T},i \in \bm{U}} {\textit{FPL}}(A_i^\mathcal{T} ,\mathcal{M}^t) . \label{eq:fpl2} 
        \end{myAlignSSS}
\end{definition}
\vspace{-2pt}
\noindent
By substituting Equation \eqref{eq:bpl1} and \eqref{eq:fpl1} into \eqref{eq:tpl_expanded}, we have
\begin{myAlignSSS}
        \textit{TPL}(A_i^\mathcal{T} ,\mathcal{M}^t)
        =
        \textit{BPL}(A_i^\mathcal{T} ,\mathcal{M}^t)
        +
        \textit{FPL}(A_i^\mathcal{T} ,\mathcal{M}^t)
        -
        \textit{PL}_0( A_i^\mathcal{T} ,\mathcal{M}^t). \label{eq:tpl_comp1}
\end{myAlignSSS}

\noindent
Since the privacy leakage is considered as the worst case among all users in the database, by Equations \eqref{eq:bpl2} and \eqref{eq:fpl2}, we have

\vspace{3pt}
\begin{myAlignSS}
        \textit{TPL}(\mathcal{M}^t)
        =
        \textit{BPL}(\mathcal{M}^t)
        +
        \textit{FPL}(\mathcal{M}^t)
        -
        \textit{PL}_0(\mathcal{M}^t). \label{eq:tpl_comp2}
\end{myAlignSS}

\noindent
Intuitively, BPL, FPL and TPL are the privacy leakage w.r.t. the adversaries {\footnotesize $ A_i^{\mathcal{T}}(P_i^B)$ }, {\footnotesize $ A_i^{\mathcal{T}}(P_i^F)$ } and {\footnotesize $ A_i^{\mathcal{T}}(P_i^B, P_i^F)$}, respectively.
In Equation \eqref{eq:tpl_comp2}, we need to minus {\footnotesize $ PL_0(\mathcal{M}^t) $} because it is counted in both BPL and FPL.
In the following, we will dive into the analysis of BPL and FPL.

\vspace{2pt}

\textbf{BPL over time.}
For BPL, we first expand and simplify Equation \eqref{eq:bpl1} by Bayesian theorem,  {\small \textit{BPL}$(A_i^\mathcal{T}, \mathcal{M}^t) $} is equal to
\begin{myAlignSSS}
        &
        \sup_{\substack{ l_i^t,{l_i^t}',  \\ \bm{r}^1,\ldots,\bm{r}^{t-1}} }  \log
        {
                {\sum_{l_i^{t-1}} \Pr ({\bm{r}^1}, \ldots ,{\bm{r}^{t-1}}|{l_i^{t-1}},D_\mathcal{K}^{t-1}) \Pr(l_i^{t-1}|{l_i^{t}})   }
                \over 
                {\sum_{{l_i^{t-1'}}} 
                        \underbrace{ \Pr ({\bm{r}^1},\ldots, {\bm{r}^{t-1}}|{l_i^{t-1}}',D_\mathcal{K}^{t-1})}_{ {\text{(i) } {\textit{BPL}}}(A_i^\mathcal{T}, \mathcal{M}^{t-1})} 
                        \underbrace{\Pr({l_i^{t-1'}}|{l_i^{t}}')}_{\text{(ii) } P_i^B} } 
        } \nonumber \\
        &+ \sup_{\substack{ l_i^t,{l_i^t}', \bm{r}^t} }  \log \frac{\Pr({\bm{r}^{t}}|{l_i^{t}},D_\mathcal{K}^{t})}{\underbrace{\Pr({\bm{r}^{t}}|{l_i^{t}}',D_\mathcal{K}^{t})}_{\text{(iii) } \textit{PL}_0(A_i^\mathcal{T},\mathcal{M}^{t})}}.  \label{eq:bpl_cal}
\end{myAlignSSS}

\noindent
We now discuss the three annotated terms in the above equation.
The first term indicates BPL at the previous time {\small $ t-1 $}, the second term is the backward temporal correlation determined by $ P_i^B $, and the third term is equal to the privacy leakage w.r.t. adversaries in traditional DP (see Definition \ref{def:pl}).
Hence, BPL at time $ t $ depends on (i) BPL at time $ t-1 $, (ii) the backward temporal correlations, and (iii) the (traditional) privacy leakage of {\small $ \mathcal{M}^t $} (which is related to the privacy budget allocated to {\small $ \mathcal{M}^t $}).
By Equation \eqref{eq:bpl_cal}, we know that 
if {\small $ t=1 $},  {\small $  \textit{BPL}(A_i^\mathcal{T},\mathcal{M}^1 ) = \textit{PL}_0(A_i,\mathcal{M}^1)$};
if {\small $ t>1 $}, we have the following, where {\small $ \mathcal{L}^B(\cdot) $} is a \textit{backward temporal privacy loss function} for calculating the accumulated privacy loss.

\begin{myAlignS}
        \textit{BPL}(A_i^\mathcal{T}, \mathcal{M}^t ) = \mathcal{L}^B \big( \textit{BPL}(A_i^\mathcal{T}, \mathcal{M}^{t-1}) \big)+\textit{PL}_0(A_i, \mathcal{M}^{t}) \label{eq:bpl_f}
\end{myAlignS}

\vspace{-3pt}
\noindent Equation \eqref{eq:bpl_f} reveals that  {\small BPL} is calculated recursively and may \textit{accumulate over time}, as shown in Example \ref{example:bpl} (Fig.\ref{fig:tpl}(a)).

\begin{example}[\textbf{BPL due to previous releases}]
        \label{example:bpl}
        Suppose that a DP mechanism {\small $ \mathcal{M}^t$} satisfies {\small $ \textit{PL}_0(\mathcal{M}^t)=0.1 $} for each time $ t\in[1,T] $, i.e., 0.1-DP at each time point.
        We now discuss BPL at each time point w.r.t. {\scriptsize $ A_i^\mathcal{T} $} with knowledge of backward temporal correlations {\scriptsize $P_{i}^B$}.
        In an extreme case, if {\scriptsize $ P_i^B  $} indicates the strongest correlation, say,  {\scriptsize $ P_i^B=\big(\begin{smallmatrix} 1&0\\ 0&1\end{smallmatrix} \big)$}, then, at time {\small $t $}, {\scriptsize $ A_i^\mathcal{T}$} 
        knows {\scriptsize $ l_i^{t} =l_i^{t-1}=\cdots=l_i^1$}, i.e., {\scriptsize $ D^{t} =D^{t-1}=\cdots=D^1$} because of  {\scriptsize $ D^t=\{l_i^t\} \cup D_\mathcal{K}^t $} for any {\scriptsize $ t \in [1,T] $}.
        Hence, the continuous data release {\scriptsize $ \bm{r}^1,\ldots,\bm{r}^t $} is equivalent to releasing the \textit{same} database multiple times; the privacy leakage at each time point will accumulate from previous time points and increase linearly (the red line with circle marker in Figure \ref{fig:tpl}(a)).
        In another extreme case, if there is no backward temporal correlation that is known to {\scriptsize $ A_i^\mathcal{T} $} (e.g., for the $ A_i $ in Definition \ref{def:pl} or {\scriptsize $ A_i^\mathcal{T}(P_i^F) $}),
        BPL at each time point is {\footnotesize $ \textit{PL}_0(\mathcal{M}^t) $}, as shown in Figure \ref{fig:tpl}(a), the black line with rectangle marker.
        The blue line with triangle marker in Figure \ref{fig:tpl}(a) depicts the backward privacy leakage caused by {\small $ P_i^B=\big(\begin{smallmatrix} 0.8&0.2\\ 0&1\end{smallmatrix} \big)$}, which can be finely quantified using our method (Algorithm \ref{algo:cal_bpl})  in Section \ref{sec:calculate_tpl}.
\end{example}

\textbf{FPL over time.}
For FPL,  similar to the analysis of BPL,  we expand and simplify Equation \eqref{eq:bpl1} by Bayesian theorem,  {\small \textit{FPL}$(A_i^\mathcal{T}, \mathcal{M}^t) $} is equal to
\vspace{5pt}
\begin{myAlignSSS}
        &
        \sup_{\substack{ l_i^t,{l_i^t}',  \\ \bm{r}^{t+1},\ldots,\bm{r}^{T}} }  \log
        {
                {\sum_{l_i^{t+1}} \Pr (\bm{r}^{t+1}, \ldots ,{\bm{r}^{T}}|{l_i^{t+1}},D_\mathcal{K}^{t+1}) \Pr(l_i^{t+1}|{l_i^{t}})   }
                \over 
                {\sum_{{l_i^{t+1'}}} 
                        \underbrace{ \Pr ({\bm{r}^{t+1}},\ldots, {\bm{r}^{T}}|{l_i^{t+1}}',D_\mathcal{K}^{t+1})}_{ {\text{(i) } {\textit{FPL}}}(A_i^\mathcal{T}, \mathcal{M}^{t+1})} 
                        \underbrace{\Pr({l_i^{t+1'}}|{l_i^{t}}')}_{\text{(ii) } P_i^F} } 
        } \nonumber \\
        &+ \sup_{\substack{ l_i^t,{l_i^t}', \bm{r}^t} }  \log \frac{\Pr({\bm{r}^{t}}|{l_i^{t}},D_\mathcal{K}^{t})}{\underbrace{\Pr({\bm{r}^{t}}|{l_i^{t}}',D_\mathcal{K}^{t})}_{\text{(iii) } \textit{PL}_0(A_i^\mathcal{T},\mathcal{M}^{t})}}.  \label{eq:fpl_cal}
\end{myAlignSSS}
\vspace{-5pt}
\noindent
By Equation \eqref{eq:fpl_cal}, we know that 
if {\small $ t=T $},  {\footnotesize $  \textit{FPL}(A_i^\mathcal{T},\mathcal{M}^T ) = \textit{PL}_0(A_i,\mathcal{M}^T)$};
if {\small $ t<T $}, we have the following, where {\small $ \mathcal{L}^F(\cdot) $} is a \textit{forward temporal privacy loss function} for calculating the increased privacy loss due to FPL at the next time.

\begin{myAlignSSS}
        \textit{FPL}(A_i^\mathcal{T}, \mathcal{M}^t ) = \mathcal{L}^F \big( \textit{FPL}(A_i^\mathcal{T}, \mathcal{M}^{t+1}) \big)+\textit{PL}_0(A_i, \mathcal{M}^{t}) \label{eq:fpl_f}
\end{myAlignSSS}

\vspace{-8pt}
\noindent Equation \eqref{eq:fpl_f} reveals that FPL is calculated recursively and may \textit{increase over time}, as shown in Example \ref{example:fpl} (Fig.\ref{fig:tpl}(b)).

\begin{example}[\textbf{FPL due to future releases}]
        \label{example:fpl}
        Considering the same setting in Example \ref{example:bpl}, we now discuss FPL at each time point w.r.t. {\footnotesize $ A_i^\mathcal{T} $} with knowledge of forward temporal correlations {\footnotesize $P_{i}^F$}.
        In an extreme case, if {\footnotesize $ P_i^F  $} indicates the strongest correlation, say,  {\footnotesize $ P_i^F=\big(\begin{smallmatrix} 1&0\\ 0&1\end{smallmatrix} \big)$}, then, at time {\small $t $}, {\footnotesize $ A_i^\mathcal{T}$} 
        knows {\footnotesize $ l_i^{t} =l_i^{t+1}=\cdots=l_i^T$}, i.e., {\footnotesize $ D^{t} =D^{t+1}=\cdots=D^T$} because of  {\footnotesize $ D^t=\{l_i^t\} \cup D_\mathcal{K}^t $} for any {\footnotesize $ t \in [1,T] $}.
        Hence, the continuous data release {\footnotesize $ \bm{r}^t,\ldots,\bm{r}^T $} is equivalent to releasing the same database multiple  times; 
        the privacy leakage at time $ t $ will increase when every time new release (i.e., {\footnotesize $ \bm{r}^{t+1} $},{\footnotesize $ \bm{r}^{t+2} $},...) happens, as shown in the red line with circle marker in Figure \ref{fig:tpl}(b).
       We see that contrary to BPL, the FPL at time 1 is the highest (due to future releases at time 1 to 10) while FPL at time 10 is the lowest (since there is no future release with respect to time 10 yet). 
        \ul{If $ \bm{r}^{11} $ is released, all FPL at time $ t \in [1,10] $ will be updated}.
        In another extreme case, if there is no forward temporal correlation that is known to {\footnotesize $ A_i^\mathcal{T} $} (e.g., for the $ A_i $ in Definition \ref{def:pl} or {\footnotesize $ A_i^\mathcal{T}(P_i^B) $}),
        then the forward privacy leakage at each time point is {\footnotesize $ \textit{PL}_0(\mathcal{M}^t) $}, as shown in the black line with rectangle marker Figure \ref{fig:tpl}(b).
		The blue line with triangle marker in Figure \ref{fig:tpl}(b) depicts the forward privacy leakage caused by {\small $ P_i^F=\big(\begin{smallmatrix} 0.8&0.2\\ 0&1\end{smallmatrix} \big)$}, which can be finely quantified using our method (Algorithm \ref{algo:cal_bpl}) in Section \ref{sec:calculate_tpl}.
\end{example}

\vspace{-10pt}
\begin{remark}
        \label{remark}
        The extreme cases shown in Examples \ref{example:bpl} and \ref{example:fpl}  are the upper and lower bound  of BPL and FPL.
        Hence, the temporal privacy loss functions {\scriptsize $ \mathcal{L}^B(\cdot) $} and {\scriptsize $ \mathcal{L}^F(\cdot) $}  in Equations \eqref{eq:bpl_f} and  \eqref{eq:fpl_f} satisfy  {\scriptsize $ 0  \leq \mathcal{L}^B(x)  \leq  x$}, where $ x $ is  BPL at the previous time, and {\scriptsize $ 0  \leq \mathcal{L}^F(x)  \leq  x$}, where $ x $ is FPL at the next time, respectively.
\end{remark}

\vspace{-5pt}

From Examples \ref{example:bpl} and \ref{example:fpl}, we know that: backward temporal correlation (i.e.,{\footnotesize $ P_i^B $}) does not affect FPL, and forward temporal correlation (i.e.,{\footnotesize $ P_i^F $}) does not affect BPL. 
In other words, adversary {\footnotesize $A_i^\mathcal{T}(P_i^B)  $} only causes BPL;  {\footnotesize $A_i^\mathcal{T}(P_i^F)  $} only causes FPL; while  {\footnotesize $A_i^\mathcal{T}(P_i^B,P_i^F)  $} poses a risk on both BPL and FPL. 
The composition of BPL and FPL is shown in Equations \eqref{eq:tpl_comp1} and \eqref{eq:tpl_comp2}.
Figure \ref{fig:tpl}(c) shows TPL, which can be calculated using Equation \eqref{eq:tpl_comp2}.


\vspace{-10pt}
\subsection{$ \alpha $-DP$ _\mathcal{T}$ and Its Composability}
\label{subsec:comp}
In this section, we define $ \alpha $-DP$ _\mathcal{T}$ (differential privacy under temporal correlations) to provide a privacy guarantee against temporal privacy leakage.
We prove its sequential composition theorem and discuss the connection between $ \alpha $-DP$ _\mathcal{T}$ and  $ \epsilon $-DP in terms of event-level/user-level privacy\cite{dwork_differential_2010-1}\cite{dwork_differential_2010} and $ w $-event privacy\cite{kellaris_differentially_2014}.

\begin{definition}[$ \alpha \text{-} DP_\mathcal{T} $, differential privacy under temporal correlations]
        \label{def:DP}
        If TPL of a differentially private mechanism  is less than or equal to $ \alpha $, we say that such mechanism satisfies $ \alpha$-Differential Privacy under Temporal correlation, i.e., {\small $ \alpha \text{-} DP_\mathcal{T} $}.
\end{definition}

\noindent
{\small DP$_\mathcal{T} $} is an enhanced version of differential privacy on time-series data.
If the data are temporally independent (i.e., for all user $ i $, both {\small $ P_i^B $} and {\small $ P_i^F $} are $ \emptyset $), an $ \epsilon $-DP  mechanism satisfies  $ \epsilon $-DP$_\mathcal{T} $.
If the data are temporally correlated (i.e., existing user $ i $ whose  {\small $ P_i^B $} or {\small $ P_i^F $} is not $ \emptyset $), an $ \epsilon $-DP  mechanism satisfies  $ {\alpha}$-DP$_\mathcal{T} $ where $ {\alpha}$ is the increased privacy leakage and can be quantified in our framework.

One may wonder, for a sequence of DP$_\mathcal{T} $ mechanisms on the timeline, what is the overall privacy guarantee.
In the following,  we suppose that {\footnotesize $\mathcal{M}^t $} is a $ \epsilon_t $-DP mechanism at time $ t\in [1,T]$  and poses risks of BPL and FPL as {\footnotesize $ \alpha_t^B $} and {\footnotesize $ \alpha_t^F $}, respectively.
That is, {\footnotesize $\mathcal{M}^t $} satisfies  {\footnotesize $ ( \alpha_t^B + \alpha_t^F - \epsilon_t )$}-DP$ _\mathcal{T} $ at time $ t $ according to Equation \eqref{eq:tpl_comp2}.
Similar to definition of TPL w.r.t. a DP mechanism at a single time point, we define TPL of a sequence of DP mechanisms at consecutive time points as follows.
\begin{definition}[TPL of a sequence of DP mechanisms]
        \label{def:combined_pl}
        The temporal privacy leakage of DP mechanisms {\footnotesize $ \{\mathcal{M}^t,\ldots,\mathcal{M}^{t+j}\} $} where $ j \geq 0 $ is defined as follows.
        \begin{myAlignSSS}
                \textit{TPL}\big( \{\mathcal{M}^t,\ldots,\mathcal{M}^{t+j}\} \big) \eqdef
                \sup_{\substack{ D^t,...,D^{t+j},\\ {D^t}',...,{D^{t+j}}',\\ \bm{r}^1, \ldots, \bm{r}^T} } \log
                {{\Pr ({\bm{r}^1}, \ldots, {\bm{r}^{T}}|{D^t,\ldots,D^{t+j}})} \over {\Pr ({\bm{r}^{1}}, \ldots ,{\bm{r}^{T}}|{D^t}',\ldots, {D^{t+j}}')}} \nonumber 
        \end{myAlignSSS}
\end{definition}
\vspace{-5pt}
It is easy to see that, if {\footnotesize $ j=0 $}, it is event-level privacy; if {\footnotesize $ t=1 $} and {\footnotesize $ j=T-1 $}, it is user-level privacy.

\begin{theorem}[Sequential Composition under Temporal Correlations]
        \label{thm:composition}
        A sequence of DP mechanisms {\footnotesize $\{\mathcal{M}^t,\ldots, \mathcal{M}^{t+j} \}$} satisfies 
        \begin{myAlignSSS}
                \begin{cases}
                        {\scriptsize (\alpha_{t}^B + \alpha_{t+1}^F )}\textit{-DP}_\mathcal{T}  & j=1 \\
                        {\scriptsize \big(\alpha_{t}^B + \alpha_{t+j}^F + \sum_{k=t+1}^{k=t+j-1}\epsilon_{k} \big)}\textit{-DP}_\mathcal{T}  &   j \geq 2 
                \end{cases}
        \end{myAlignSSS}
\end{theorem}


\noindent
In Theorem \ref{thm:composition}, when {\scriptsize $ t=1 $} and {\scriptsize $ j=T-1 $}, we have   Corollary \ref{col:user-level}.
\begin{corollary}
        \label{col:user-level}
        The temporal privacy leakage of a combined mechanism {\small $\{\mathcal{M}^1,\ldots, \mathcal{M}^{T} \}$} is  {\footnotesize $ \sum_{k=1}^{k=T}\epsilon_{k} $} where $\epsilon_{k}  $ is the privacy budget of $ \mathcal{M}^{k} $, $ k\in[1,T] $. 
\end{corollary}
\noindent
It shows that \textit{temporal correlations do not affect the user-level privacy} (i.e., protecting all the data on the timeline of each user), which is in line with the idea of group differential privacy: protecting all the correlated data in a bundle.

\vspace{2pt}
\textbf{Comparison  between DP and DP$ _\mathcal{T} $.}
As we mentioned in Section \ref{sec:problem_setting}, there are typically two privacy notions in continuous data release: \textit{event-level} and \textit{user-level}
\cite{dwork_differential_2010-1}
\cite{dwork_differential_2010}.
Recently, $ w $-event privacy\cite{kellaris_differentially_2014} is proposed to merge the gap between event-level and user-level privacy.
It protects the data in any $ w $-length sliding window by utilizing the  sequential composition theorem of DP.

\begin{theorem}[Sequential composition on independent data\cite{mcsherry_privacy_2009}]
        \label{thm:sequential_composition}
        Suppose that {\footnotesize $ \mathcal{M}^t $} satisfies {\footnotesize $ \epsilon_t $}-DP for each {\footnotesize $ t\in[1,T] $}.
        A combined mechanism {\footnotesize $\{\mathcal{M}^t,\ldots, \mathcal{M}^{t+j} \}$} satisfies {\footnotesize $ \sum_{k=t}^{k=t+j} \epsilon_{k} $}-DP.
\end{theorem}

For ease of exposition, suppose that {\footnotesize $ \mathcal{M}^t $} satisfies {\footnotesize $ \epsilon $}-DP for each {\footnotesize $ t \in[1,T]$}.
According to Theorem \ref{thm:sequential_composition}, it achieves $ T\epsilon $-DP on user-level and $ w\epsilon $-DP on $ w $-event level.
We compare the  privacy guarantee of $ \mathcal{M}^t $ on independent data and temporally correlated data in Table \ref{tb:composition}.

\renewcommand{\arraystretch}{1.25}
\begin{table}[h]
        \centering
        \footnotesize
        \vspace{-10pt}
        \caption{The privacy guarantee of $ \epsilon $-DP mechanisms on different types of data.}
         \vspace{-10pt}
        \begin{tabularx}{8.6cm}{|p{2.8cm}|p{1.4cm}|X|}
                \hline  
                \diagbox[width=3.2cm]{{Privacy Notion}}{{Data}} &  independent &  temporally correlated  \\
                \hline event-level{ \cite{dwork_differential_2010-1}\cite{dwork_differential_2010}}& $ \epsilon $-DP &  ${\alpha} $-DP$ _\mathcal{T}$ ($\epsilon \leq \alpha \leq T\epsilon $)  \\
                \hline $ w $-event\cite{kellaris_differentially_2014} & $w \epsilon $-DP & see Theorem \ref{thm:composition}  \\
                \hline user-level{ \cite{dwork_differential_2010-1}\cite{dwork_differential_2010}} & $ T\epsilon $-DP & $ T\epsilon $-DP$ _\mathcal{T} $ (by Corollary \ref{col:user-level})\\ 
                \hline 
        \end{tabularx}
        \label{tb:diff} 
        \label{tb:composition}
\end{table}

It reveals that \textit{temporal correlations may blur the boundary between event-level privacy and user-level privacy}.
As shown in Table \ref{tb:composition}, the  privacy leakage of a $ \epsilon $-DP mechanism at a single time point (event-level) on temporally correlated data may range from $ \epsilon $ to $ T\epsilon $ which depends on the strength of temporal correlation.
For example, as shown in Figure \ref{fig:tpl}(c), TPL of a $ 0.1 $-DP mechanism under strong temporal correlation at each time point (event-level) is $  T*0.1$, which is equal to  the privacy leakage of a sequence of $ 0.1 $-DP mechanisms that satisfies user-level privacy.
When the temporal correlations is moderate, TPL of a $ 0.1 $-DP mechanism  at each time point (event-level) is less than $  T*0.1$ but still larger than $ 0.1 $, as shown in the blue line with triangle markers in Figure \ref{fig:tpl}(c).

%



\cyrev{
	\subsection{Connection with Pufferfish Framework}
	\label{subsec:pufferfish}
	$ \alpha $-DP$_\mathcal{T} $ is highly related to Pufferfish framework \cite{kifer_rigorous_2012}\cite{kifer_pufferfish:_2014} and its instantiation \cite{song_pufferfish_2017}.
	Pufferfish provides rigorous and customizable privacy framework which consists of three components: a set of secrets, a set of discriminative pairs and data evolution scenarios (how the data were generated, or how the data are correlated).
	Note that the data evolution scenarios is essentially the adversary's background knowledge about data, such as data correlations.
	 In Pufferfish framework \cite{kifer_rigorous_2012} \cite{kifer_pufferfish:_2014}, they prove that when secrets to be all possible values of tuples, discriminative pairs to be the set of all pairs of potential secrets, and tuples to be independent, it is equivalent to DP; hence, Pufferfish becomes $DP_\mathcal{T} $  under the above setting of secrets and discriminative pairs, but with temporally correlated tuples.}

\cyrev{	
A recent work  \cite{song_pufferfish_2017} proposes Markov Quilt Mechanism for Pufferfish framework when  data correlations can be modeled by Bayesian Network (BN).
They further design efficient mechanisms MQMExact and MQMApprox when the structure of BN is Markov Chain.
Although we also use Markov Chain as temporal correlation model, the settings of two studies are essentially different.
In the motivating example of \cite{song_pufferfish_2017}, i.e., Physical Activity Monitoring, the private histogram is ``one-shot'' (see Section \ref{sec:problem_setting}) release:
adapting to our setting in Figure \ref{fig:corr_example}, their example is equivalent to release location access statistics for a one-user database only at a given time $ T $.
Whereas, we focus on quantifying the privacy leakage in continuous data release.
One important insight of our study is that, \textit{the privacy guarantee of one-shot data release is not on its own or static; instead, it may be affected by previous release and even future release}, which are defined as Backward and Forward Privacy Leakage in our work.
}

\vspace{-8pt}
\subsection{Discussion}
\label{subsec:discussion}
We make a few important observations regarding our privacy analysis.
First, the temporal privacy leakage is defined in a personalized way. 
That is, the privacy leakage may be different for users with distinct temporal patterns (i.e., {\footnotesize $ P_i^B $} and {\footnotesize $ P_i^F $}).
We define the overall temporal privacy leakage as the maximum one for all users, so that $ \alpha $-DP$ _\mathcal{T} $ is compatible with the traditional $ \epsilon $-DP mechanisms (using one parameter to represent the overall privacy level) and we can convert them  to protect against TPL.
On the other hand, our definitions is also compatible with personalized differential privacy (PDP) mechanisms\cite{jorgensen_conservative_2015}, in which the personalized privacy budgets, i.e., a vector {\footnotesize $ [\epsilon_1, \ldots,\epsilon_n]$}, are allocated to each user.
In other words, we can convert a PDP mechanism to bound the temporal privacy leakage for each user.

Second, in this paper, we focus on the temporally correlated data and assume that the adversary has knowledge of temporal correlations modeled by Markov chain.
However, it is possible that the adversary has knowledge about more sophisticated temporal correlation model or other types of correlations.
\cyrev{Especially, if the the assumption about Markov model is not the ``ground truth''.
We may not protect against TPL appropriately.
Song et al. \cite{song_pufferfish_2017} provided a formalization of the upper bound of privacy leakage if a set of possible data correlations $  \Theta $ is given, which is \textit{max-divergence} between any two conditional distributions of $ \tilde{\theta} \in \Theta $  and $ \theta \in \Theta$ given any secret that we want to protect.
In their Markov Quilt Mechanism that models correlation using Bayesian Network (BN), the privacy leakage can be represented by \textit{max-influence} between nodes (tuples).
Hence, given a set of  BNs as possible data correlations, the upper bound of privacy leakage is the largest max-influence under different BNs.
Similarly, in our case, given a set of Transition Matrices (TM), we can calculate the maximum TPL w.r.t. different TMs.
However, it remains an open question how to find a set of  possible data correlations $  \Theta $ that includes the adversarial knowledge or the ground truth in a high probability,  which may depend on the application scenarios.
We believe that this question is also related to ``how to identify appropriate concepts of neighboring databases in different scenarios'' since properly defined neighboring databases can circumvent the affect of data correlations on privacy leakage in a certain extent (as we show in Section \ref{subsec:comp},  the difference between event-level privacy, w-event privacy, and user-level privacy). 
}

\section{Calculating Temporal Privacy Leakage}
\label{sec:calculate_tpl}
In this section, we design algorithms for computing backward privacy leakage (BPL) and forward privacy leakage (FPL).
We first show that both of them  can be transformed to the optimal solution of a \textit{linear-fractional programming} problem\cite{bajalinov_linear-fractional_2003} in Section \ref{sec:formulation}.
Traditionally, this type of problem can be solved by simplex algorithm in exponential time \cite{bajalinov_linear-fractional_2003}.
By exploiting the constraints in this problem, we then design a polynomial algorithm to calculate it in Section \ref{sec:polynomial}.
Further, based on our observation of some repeated computation  when the polynomial algorithm runs on different time points, we precompute such  common results and design quasi-quadratic and sub-linear algorithms for calculating temporal privacy leakage at each time point in Section \ref{sec:log-linear} and  Section \ref{sec:sub-linear}, respectively.

\vspace{-5pt}
\subsection{Problem formulation}
\label{sec:formulation}
According to the privacy analysis of BPL and FPL in Section \ref{subsec:tpl}, we need to solve the backward  and forward temporal privacy loss functions $ \mathcal{L}^B(\cdot) $ and $ \mathcal{L}^F (\cdot)$ in Equations \eqref{eq:bpl_f} and \eqref{eq:fpl_f}, respectively.
By observing the structure of the first term in Equations \eqref{eq:bpl_cal} and \eqref{eq:fpl_cal}, we can see that the calculations for recursive functions $ \mathcal{L}^B (\cdot)$ and $ \mathcal{L}^F (\cdot)$ are virtually in the same way.
They calculate the increment of the input values (the previous BPL or the next FPL) based on temporal correlations (backward or forward).
Although different degree of correlations result in different privacy loss functions, the methods for analyzing them are the same.

We now quantitatively analyze the temporal privacy leakage.
In the following, we demonstrate the calculation of $ \mathcal{L}^B (\cdot)$, i.e., the first term of Equation \eqref{eq:bpl_cal}.
\begin{myAlignSSS}
        \sup_{\substack{ l_i^t,{l_i^t}',  \\ \bm{r}^1,\ldots,\bm{r}^{t-1}} }  \log
        {
                {\sum_{l_i^{t-1}} \Pr ({\bm{r}^1}, \ldots ,{\bm{r}^{t-1}}|{l_i^{t-1}},D_\mathcal{K}^{t-1}) \Pr(l_i^{t-1}|{l_i^{t}})   }
                \over 
                {\sum_{{l_i^{t-1'}}} 
                        \underbrace{ \Pr ({\bm{r}^1},\ldots, {\bm{r}^{t-1}}|{l_i^{t-1}}',D_\mathcal{K}^{t-1})}_{ { {\textit{BPL}}}(A_i^\mathcal{T}, \mathcal{M}^{t-1})} 
                        \underbrace{\Pr({l_i^{t-1'}}|{l_i^{t}}')}_{P_i^B} } 
        } \label{eq:bpl_cal_1st}
\end{myAlignSSS}

We now simplify the notations in the above formula.
Let two arbitrary (distinct) rows in $ P_i^B $ be vectors {\footnotesize $ \bm{q}=(q_1,...,q_n)$} and {\footnotesize $\bm{d}=(d_1,...,d_{n}) $}.
For example, suppose that {\small $ \bm{q} $} is the first row in the transition matrix of Figure \ref{fig:mc}(b);  then, the elements in {\small $ \bm{q} $} are:  {\footnotesize $q_1= \Pr(l_i^{t-1}=loc_1|l_i^t=loc_1) $}, {\footnotesize $q_2= \Pr(l_i^{t-1}=loc_2|l_i^t=loc_1) $}, {\footnotesize $q_3= \Pr(l_i^{t-1}=loc_3|l_i^t=loc_1) $}, etc.
Let {\footnotesize $ \bm{x}=(x_1,...,x_{n})^\mathrm{T}$} be a vector whose elements indicate {\footnotesize $\Pr ({\bm{r}^1},... ,{\bm{r}^{t-1}}|l_i^{t-1}, {D_{\mathcal{K}}^{t-1}})$}  with distinct values of {\small $ l_i^{t-1} \in \bm{loc} $}, e.g., {\footnotesize $ x_1 $} denotes {\footnotesize $\Pr ({\bm{r}^1},... ,{\bm{r}^{t-1}}|{l_i^{t-1}}=loc_1, {D_{\mathcal{K}}^{t-1}})$}.
We obtain the following by expanding {\small $ l_i^{t-1},l_i^{t-1'} \in \bm{loc} $} in Equation \eqref{eq:bpl_cal_1st}.

\begin{myAlignSSS}
        \mathcal{L}^B \big( \textit{BPL}(A_i^\mathcal{T}, \mathcal{M}^{t-1})\big) 
        &= \sup_{ \bm{q}, \bm{d} \in   P_i^B } \log \frac{q_1x_1+\cdots+q_n x_n}{d_1 x_1+\cdots+d_n x_n} \nonumber\\
        &= \sup_{ \bm{q}, \bm{d} \in   P_i^B } \log  \frac{ {\boldsymbol{q}  \boldsymbol{x}}}{{\boldsymbol{d}   \boldsymbol{x}}} \label{eq:lfp_M}
\end{myAlignSSS}

Next, we formalize the objective function and constraints. 
Suppose that {\small $ \textit{BPL}(A_i^\mathcal{T}, \mathcal{M}^{t-1}) ={{\alpha}^B_{t-1}}$}.
According to the definition of BPL (as the supremum), for any {\footnotesize $ x_j,x_k \in \bm{x} $}, we have {\footnotesize $e^{-{{\alpha}^B_{t-1}}}  \leq \frac{x_j}{x_k} \leq e^{{{\alpha}^B_{t-1}}}$}.
Given $ \bm{x} $ as the variable vector and $ \bm{q},\bm{d} $ as the coefficient vectors,  {\footnotesize $ \mathcal{L}^B ({{\alpha}^B_{t-1}}) $}  is equal to the logarithm of the objective function \eqref{eq:lfp} in the following maximization problem.
\begin{myAlignSSS}
        \textnormal{maximize }  & \text{ }    \frac{ {\boldsymbol{q}  \boldsymbol{x}}}{{\boldsymbol{d}   \boldsymbol{x}}}  \label{eq:lfp}\\
        \textnormal{subject to  } 
        & \text{ } e^{-{{\alpha}^B_{t-1}}} \leq \frac{x_j}{x_k} \leq e^{{{\alpha}^B_{t-1}}},   \label{eq:lfp_end} \\
        &\text{ }  0 < {x_j}< 1 \text{ and }  0<x_k<1, \label{eq:lfp_end2}  \\
        & \text{ where } x_j,x_k \in \bm{x}, \text{ } j,k \in  [1,n]. \nonumber
\end{myAlignSSS}

\vspace{-10pt}
The above is a form of \textit{Linear-Fractional Programming}\cite{bajalinov_linear-fractional_2003} (i.e., LFP), where the objective function is a ratio of two linear functions and the constraints are linear inequalities or equations.
A linear-fractional programming problem can be converted into a sequence of linear programming problems and then solved using the simplex algorithm in time {\small $ O(2^{n}) $} \cite{bajalinov_linear-fractional_2003}.
According to Equation \eqref{eq:lfp_M}, given $ P_i^B $ as an $  n\times n $ matrix, finding $ \mathcal{L}^B(\cdot) $ involves solving $ n(n-1) $ such LFP problems w.r.t. different permutation of choosing two rows $ \boldsymbol{q} $ and $ \boldsymbol{d} $ from the transition matrix $ P_i^B $.
Hence, the overal time complexity using simplex algorithm is $ O(n^2 2^n) $.

As we mentioned previously, the calculations of $ \mathcal{L}^B(\cdot) $ and  $ \mathcal{L}^F(\cdot) $ are identical.
For simplicity, in the following part of this paper, we use $ \mathcal{L(\cdot)} $ to represent the privacy loss function $ \mathcal{L}^B(\cdot) $ or $ \mathcal{L}^F(\cdot) $,  use $ \alpha $ to denote $ \alpha_{t-1}^B $ or $ \alpha_{t+1}^F $, and use $ P_i $ to denote $  P_i^B $ or $  P_i^F$.

\vspace{-15pt}
\subsection{Polynomial Algorithm}
\label{sec:polynomial}
\cyrev{
In this section, we design an efficient algorithm to calculate TPL, i.e., to solve the linear-fractional program in Equations \eqref{eq:lfp} to \eqref{eq:lfp_end2}.
Intuitively, we prove that the optimal solutions always satisfy some conditions (Theorem \ref{thm:lfp}), which enable us to design an efficient algorithm to obtain the optimal value without directly solving the optimization problem.
}

\textbf{Properties of the optimal solutions.}
From Inequalities \eqref{eq:lfp_end} and \eqref{eq:lfp_end2}, we know that the feasible region of the constraints are not empty and bounded; hence, the optimal solution exists.
We prove Theorem \ref{thm:lfp}, which enables the optimal solution to be found in time {\small $ O(n^2) $}.

We first define some notations that will be frequently used  in our theorems.
Suppose that the variable vector $ \bm{x} $ consists of two parts (subsets):  $ \bm{x}^+$ and $ \bm{x}^-$.
Let the corresponding coefficients vectors be {\small $ \bm{q}^+, \bm{d}^+ $} and {\small $ \bm{q}^-, \bm{d}^-$}.
Let {\small $ q=\sum{\bm{q}^+} $} and {\small $ d=\sum{\bm{d}^+} $}.
For example, suppose that  {\small $ \bm{x}^+=[x_1,x_3] $} and {\small $ \bm{x}^-=[x_2,x_4,x_5] $}. 
Then, we have {\small $ \bm{q}^+=[q_1,q_3] $}, {\small $ \bm{d}^+=[d_1,d_3] $}, {\small $ \bm{q}^-=[q_2,q_4,q_5] $}, and {\small $ \bm{d}^-=[d_2,d_4,q_5] $}. In this case, {\small $ q=q_1+q_3 $} and {\small $ d=d_1+d_3 $}.

\vspace{2pt}
\begin{theorem}
	\label{thm:lfp}
	If  the following Inequalities \eqref{eq:cond1} and \eqref{eq:cond2} are satisfied, the maximum value of the objective function in the problem \eqref{eq:lfp}$ \sim $\eqref{eq:lfp_end2} is {\footnotesize $ \frac{q(e^{{\alpha}_{t-1}^B} -1) + 1 }{d(e^{{\alpha}_{t-1}^B} - 1) + 1} $}.
	\begin{myAlignSS}
		&\frac{q_j}{d_j} > \frac{q(e^{{\alpha}_{t-1}^B} -1) + 1 }{d(e^{{\alpha}_{t-1}^B} - 1) + 1}, & \forall j\in[1,n] \text{ where } q_j \in \bm{q}^+, d_j \in \bm{d}^+ \label{eq:cond1}\\
		&\frac{q_k}{d_k} \leq \frac{q(e^{{\alpha}_{t-1}^B} -1) + 1 }{d(e^{{\alpha}_{t-1}^B} - 1) + 1}, & \forall k\in[1,n] \text{ where } q_k \in \bm{q}^-, d_k \in \bm{d}^- \label{eq:cond2}
	\end{myAlignSS}
\end{theorem}


According to Equation \eqref{eq:lfp_M}, the increment of  temporal privacy loss, i.e., $ \mathcal{L}(\cdot) $, is the maximum value among the $ n(n-1) $ LFP problems, which are defined by different 2-permutations  $ \bm{q} $ and $ \bm{d} $ chosen from $ n $ rows of $ P_i$.

\begin{myAlignS}
	&\mathcal{L}
	\big( {{\alpha}} \big)  = \max_{ \bm{q}, \bm{d} \in   P_i } 
	\log\frac{ q(e^{{\alpha}}-1) +1}{ d(e^{{\alpha}}-1) +1} \label{eq:lfp_sol_b}  
\end{myAlignS}

Further, we give Corollary \ref{col:q_d} for finding {\small $ \bm{q}^+ $} and {\small $ \bm{d}^+ $}.
\begin{corollary}
	\label{col:q_d}
	If Inequalities \eqref{eq:cond1} and \eqref{eq:cond2} are satisfied, we have {\small $ q_j > d_j $} in which $ q_j \in \bm{q}^+ $ and  {\small $ d_j \in \bm{d}^+ $}.
\end{corollary}


Now, the question is how do we find $ q $ and $ d $ (or, {\small $ \bm{q}^+ $} and {\small $ \bm{d}^+ $}) in Theorem \ref{thm:lfp}  that give the maximum value of objective function.
Inequalities \eqref{eq:cond1} and \eqref{eq:cond2} are sufficient conditions for obtaining such optimal value.
Corollary \ref{col:q_d} gives a necessary condition for satisfying Inequalities \eqref{eq:cond1} and \eqref{eq:cond2}.
Based on the above analysis, we design Algorithm \ref{algo:cal_bpl} for computing BPL or FPL.

\begin{algorithm}[h]
        \begin{scriptsize}
                \SetKwRepeat{doWhile}{do}{while}
                \SetKwComment{tcc}{//}{}
                \DontPrintSemicolon 
                \KwIn{$ P_i $  {\scriptsize (i.e., $  P_i^B$ or $  P_i^F$)};  $ \alpha $ (i.e., {\scriptsize ${ {{\alpha}^B_{t-1}}}$ or ${ {{\alpha}^F_{t+1}}}$}); $ \epsilon_t $ ({\scriptsize i.e., $\textit{PL}_0(\mathcal{M}^t) $}). }
                \KwOut{BPL or FPL at time $ t $}
                $ \mathcal{L} = 0$;  \tcc*{\scriptsize the value of Equation \eqref{eq:lfp_sol_b}.} 
                \ForEach{two rows permutation $ \bm{q} $, $ \bm{d}  $ from $ n $ rows in $ P_i $ }
                {  
                        \ForEach(\tcc*[f]{\scriptsize Corollary \ref{col:q_d}.}){$ q_j \in\bm{q},d_j \in\bm{d} $}    
                        { \lIf{$  q_j>d_j $} 
                                {add $ q_j $ to $ \bm{q}^+$; add $ d_j $ to $ \bm{d}^+  $; } 
                        }
                        $\textsf{\textit{update}} = false ;$\;
                        \doWhile(\tcc*[f]{\scriptsize find $ \bm{q}^+ $, $ \bm{d}^+ $ by Theorem \ref{thm:lfp}.} ) {\textsf{update}}{  
                                {\scriptsize $q = \sum{\bm{q}^+}  $};  {\scriptsize $ d = \sum{\bm{d}^+}; $} \tcc*{{\scriptsize update $ q $ and $ d $}.}
                                \ForEach{$ q_j \in\bm{q}^+,d_j \in\bm{d}^+ $}
                                {  \tcc{{\scriptsize if it does not satisfy Inequality \eqref{eq:cond1}.}} 
                                        \lIf{{\scriptsize $  q_j/d_j \leq \big(q*(e^{{\alpha}}-1)+1 \big)/\big(d*(e^{{\alpha}}-1)+1 \big) $} \label{algo:iterate} \;} 
                                        {\scriptsize    $ \bm{q}^+ \leftarrow \bm{q}^+ - q_j $;
                                                $ \bm{d}^+ \leftarrow \bm{d}^+ - d_j $;
                                                $\textsf{\textit{update}} = true $;
                                                \label{algo:findtpl:update}}  
                                }
                        }
                        \lIf{\scriptsize $ \mathcal{L}    <  \log \frac{q*(e^{{\alpha}} -1) +1}{d*(e^{{\alpha}} -1) +1} $  }
                        {\scriptsize $ \mathcal{L}    = \log \frac{q*(e^{{\alpha}} -1) +1}{d*(e^{{\alpha}} -1) +1} $;}
                }
                \Return{ {\scriptsize $  \mathcal{L} + \epsilon_t$}}   \tcc*{{\scriptsize by Equation \eqref{eq:bpl_f} or \eqref{eq:fpl_f}}.}
                \caption{Calculate BPL or FPL by Theorem \ref{thm:lfp}.}
                \label{algo:cal_bpl}
        \end{scriptsize}
\end{algorithm}

\textbf{Algorithm design}.
According to the definition of BPL and FPL, we need to find the maximum privacy leakage (Line 12) w.r.t. any 2-permutations selected from $ n $ rows of the given transition matrix $  P_i$ (Line 2).
Lines 3$ \sim $11 are to solve a single linear-fractional programming problem w.r.t two specific rows chosen from  $  P_i$.
In Lines 3 and 4, we divide the variable vector $ \bm{x} $ into two parts according to Corollary \ref{col:q_d}, which gives the necessary condition for finding the maximum solution:
if  a pair of coefficients with the same subscript {\footnotesize $ q_j \leq d_j $}, they are \textit{not} in $ \bm{q}^+ $ and $ \bm{d}^+ $ that satisfy Inequalities \eqref{eq:cond1} and \eqref{eq:cond2}.
In other words, if {\footnotesize $ q_j > d_j $}, they are ``candidates'' in $ \bm{q}^+ $ and $ \bm{d}^+ $ that gives the maximum objective function.
In Lines 5$ \sim $11, we check the candidates  $ \bm{q}^+ $ and $ \bm{d}^+ $  whether or not satisfying Inequalities \eqref{eq:cond1} and \eqref{eq:cond2}.
In Line 7, we update the values of $ q $ and $ d $ because they may be changed in the context. 
In Lines 8$ \sim $10, we check each bit in $ \bm{q}^+ $ and $ \bm{d}^+ $  whether or not satisfying Inequality \eqref{eq:cond1}.
If any bit is removed, we set a flag  \textsf{updated}  to $ true $ and do the loop again until every pairs of $ \bm{q}^+ $ and $ \bm{d}^+ $ satisfy Inequality \eqref{eq:cond1}.

A subtle question may arise regarding such ``update''.
In Lines 8$ \sim $10, the algorithm may remove \textit{several} pairs of $ q_j$ and $ d_j $, say, {\footnotesize $ \{q_{1}, d_{1}\} $} and {\footnotesize $ \{q_{2}, d_{2}\} $}, that do not satisfy Inequality \eqref{eq:cond1} in one loop.
However, one may wonder if it is possible that, after removing {\footnotesize $ \{q_{1}, d_{1}\} $} from {\footnotesize $ \bm{q}^+ $} and $ \bm{d}^+ $, Inequality \eqref{eq:cond1} can be satisfied for {\footnotesize $ \{q_{2}, d_{2}\} $} due to the update of $ q $ and $ d $, i.e., {\footnotesize $ \frac{q_2}{d_2} > \frac{(q-q_1)*(e^{{\alpha}} -1) +1}{(d-d1)*(e^{{\alpha}} -1) +1} $}.
We show that this is impossible.
If {\footnotesize $ \frac{q_1}{d_1} \leq \frac{q*(e^{{\alpha}} -1) +1}{d*(e^{{\alpha}} -1) +1}$},
we have  {\footnotesize $ \frac{q*(e^{{\alpha}} -1) +1}{d*(e^{{\alpha}} -1) +1} \leq \frac{(q-q_1)*(e^{{\alpha}} -1) +1}{(d-d1)*(e^{{\alpha}} -1) +1} $}.
Hence, {\footnotesize $ \frac{q_2}{d_2} \leq \frac{q*(e^{{\alpha}} -1) +1}{d*(e^{{\alpha}} -1) +1} \leq \frac{(q-q_1)*(e^{{\alpha}} -1) +1}{(d-d1)*(e^{{\alpha}} -1) +1} $}.
Therefore, we can remove  multiple pairs of {\footnotesize $ q_j$} and {\footnotesize $ d_j $} that do not satisfy Inequality \eqref{eq:cond1} at one time.

 Theorems \ref{thm:algo_extreme1} and \ref{thm:algo_extreme2} provide insights on transition matrices that lead to the extreme cases of  temporal privacy leakage, which are  in accordance with Remark \ref{remark}.


\begin{theorem}
 \label{thm:algo_extreme1}
If  for any two rows $  \bm{q }$ and $  \bm{d}$  chosen from $ P_i $ that satisfy $ q_i=d_i $ for $ i\in[1,n] $, we have $\mathcal{L}(\cdot)  = 0$.
\end{theorem}

\begin{theorem}
\label{thm:algo_extreme2}
If  there exist two rows $  \bm{q }$ and $  \bm{d}$  in $ P_i $ that satisfy $ q_i=1 $ and $ d_i=0$ for a certain index $ i $, $\mathcal{L}(\cdot) $ in an identical function, i.e, $ \mathcal{L}(x)=x  $.
\end{theorem}


\textbf{Complexity}.
The time complexity for solving one linear-fractional programming problem (Lines 3$ \sim $11) w.r.t. two specific rows of the transition matrix is {\small $ O(n^2) $} because Line \ref{algo:iterate} may iterate {\small $ n(n-1) $} times in the worst case.
The overall time complexity of Algorithm \ref{algo:cal_bpl}  is {\small $ O(n^4) $} since there are $ n(n-1) $ permutations of different pairs of $ \bm{q} $ and $ \bm{d} $.


\subsection{Quasi-quadratic Algorithm}
\label{sec:log-linear}
When Algorithm \ref{algo:cal_bpl} runs on different time points for continuous data release, we have a constant $ P_i $ and different $ \alpha $ as inputs at each time point.
Our observation is that, there may exist some common computations when Algorithm  \ref{algo:cal_bpl} takes different inputs of $ \alpha $.
More specifically, we want to know that,  for  given  $ \bm{q} $ and $ \bm{d} $ which are two rows chosen from $ P_i $, whether or not the final $ q $ and $ d $ (when stopping update in Line 11 in Algorithm  \ref{algo:cal_bpl})  keep the same for any input of $ \alpha $,  so that we can precompute $ q $ and $ d $ and do not need Lines 5$ \sim $11 at next run with a different $ \alpha $.
Since  $  \bm{q},\bm{d}$ and $ \alpha $ indicate a specific LFP problem in Equation \eqref{eq:lfp}$ \sim $\eqref{eq:lfp_end2}, we attempt to directly obtain the $ q $ and $ p $ in Theorem \ref{thm:lfp}.
Unfortunately, we find that such $ q $ and $ d $ do not keep the same for different $ \alpha $ w.r.t. given $ \bm{q } $ and $ \bm{d } $.
However, interestingly we find that, for any input $ \alpha $, there are only several possible pairs of $ q $ and $ d $ when the update in Lines 6$ \sim $11 of Algorithm \ref{algo:cal_bpl} is terminated, as shown in Theorem \ref{thm:lfp_k_case}.

%

\begin{theorem}
	\label{thm:lfp_k_case}
	Let $ \bm{q} $ and $ \bm{d} $ be two vectors drawn from rows of a transition matrix.
	Assume $ q_i\neq d_i, i \in [1,n]$, for each $ q_i \in \bm{q}$ and $ d_i \in \bm{d}$.\footnote{The case of $ q_i = d_i $ is proven in Theorem \ref{thm:algo_extreme1}.}
	Without loss of generality, we assume {\footnotesize $ \frac{q_1}{d_1} > \cdots > \frac{q_k}{d_k} > 1 \geq  \frac{q_{k+1}}{d_{k+1}} >  \cdots > \frac{q_n}{d_n} $}, then there are only $ k $ pairs of $ q $ and $ d $ that satisfy Inequalities \eqref{eq:cond1} and \eqref{eq:cond2} for given $ \bm{q} $ and $ \bm{d} $.
	\begin{myAlignSSS}
		\text{case 1:  }&  \frac{q_k}{d_k} > \frac{q(e^\alpha -1)+1}{d(e^\alpha-1)+1} \geq 1 \text{ where } q= \sum\nolimits_{i=1}^{k}q_i, d= \sum\nolimits_{i=1}^{k}d_i; \nonumber \\
		\text{case 2:  }&  \frac{q_{k-1}}{d_{k-1}}> \frac{q(e^\alpha-1)+1}{d(e^\alpha-1)+1}  \geq \frac{q_k}{d_k}  \text{ where } q= \sum\nolimits_{i=1}^{k-1}q_i, d= \sum\nolimits_{i=1}^{k-1}d_i; \nonumber \\
		& \vdotswithin{=}  \vdotswithin{=--------------}   \vdotswithin{=------} \notag \\
		\text{case k:  }& \frac{q_1}{d_1} > \frac{q(e^\alpha-1)+1}{d(e^\alpha-1)+1} \geq \frac{q_2}{d_2}  \text{ where } q= q_1, d=d_1. \nonumber
	\end{myAlignSSS}
\end{theorem}

More interestingly, we find that the values of $ q $ and $ d $ are monotonically decreasing with the increase of $ \alpha $. 
That is, when $ \alpha $ is increasing from $ 0 $ to $ \infty $,  the pairs of  $ q $ and $ d $ is transiting from case $ 1 $ to case $ 2 $,..., until to case $ k $.
In other words, the final $ q $ and $ d $ is constant in a certain range of $ \alpha $.
Hence, the mapping from a given $ \alpha $ to the optimal solution of a LFP problem w.r.t. $ \bm{q} $ and $ \bm{d} $ can be represented by a piecewise function as shown in Theorem  \ref{thm:lfp_seg_func}.

\begin{theorem}
	\label{thm:lfp_seg_func}
We can represent the function of the optimal solution of LPF problem w.r.t. given $ \bm{q} $ and $ \bm{d} $   by a piecewise function as follows, where $ q_i $ and $ d_i $ are the i-th elements of $ \bm{q} $ and $ \bm{d} $, respectively; and  {\scriptsize $ \alpha_j = \log \left( \frac{q_{k-j+1} - d_{k-j+1}}{q^{k-j+1}* d_{k-j+1} - d^{k-j+1}* q_{k-j+1}}  +1 \right)$} for $ j \in [1,k-1] $.
We call {\scriptsize $ \alpha_1,...,\alpha_{k-1} $} in the sub-domains as transition points; {\scriptsize $ q^1,\cdots,q^k $} and {\scriptsize $ d^1,\cdots,d^k $} as coefficients of the piecewise function.
	\begin{myAlignSSS}\label{eq:piecewise}
	&f_{\bm{q},\bm{d}}(\alpha)=   \frac{q(e^\alpha-1)+1}{d(e^\alpha-1)+1} \text{ where \space\space}  \nonumber \\  
	&	q=
	\begin{cases}
		q^k = \sum\nolimits_{i=1}^{k}q_i  \\
		q^{k-1} = \sum\nolimits_{i=1}^{k-1}q_i  \\
		\vdotswithin{=-----}     \\
		q^1 = q_1 .
	\end{cases}
	d=
	\begin{cases}
		d^k = \sum\nolimits_{i=1}^{k}d_i & 0 \leq \alpha <\alpha_1; \\
		d^{k-1}=\sum\nolimits_{i=1}^{k-1}d_i &  \alpha_1 \leq \alpha <\alpha_2; \\
		\vdotswithin{=-----} & \vdotswithin{=---}     \\
		d^1  =d_1 &  \alpha_{k-1} \leq \alpha .
	\end{cases}
\end{myAlignSSS}
\end{theorem}

\vspace{-10pt}
For convenience, we let   {\scriptsize $qArr=[ q^1, ..., q^k ] $}, 
{\scriptsize $dArr=[ d^1, ...,   d^{k}] $}, and
{\scriptsize $aArr=[ \alpha_{k-1}, ...,  0 ] $}.
Then,  \textit{qArr, dArr, aArr}  determine  a piecewise function in Equation \eqref{eq:piecewise}.
Also, we let $ qM $, $ dM $ and $ aM $ be $ n(n-1) \times n $ matrices in which  rows are \textit{qArr}, \textit{dArr} and \textit{aArr} w.r.t. distinct 2-permutations of $ \bm{q} $ and $ \bm{d} $ from $ n $ rows of $ P_i $. 
In other words, the three matrices determine $ n(n-1) $ piecewise functions.

In the following, we first design Algorithm \ref{algo:precomp1} to obtain $ qM $, $ dM $ and $ aM $.
We then design Algorithm \ref{algo:calc_precomp} to utilize the precomputed $ qM, dM, aM $ for calculating backward or forward temporal privacy leakage at each time point.

\vspace{-10pt}
\begin{algorithm}[h]
	\begin{scriptsize}
		\SetKwRepeat{doWhile}{do}{while}
		\SetKwComment{tcc}{//}{}
		\DontPrintSemicolon 
		\KwIn{$ P_i $ (i.e., $  P_i^B$ or $  P_i^F$)}
		\KwOut{matrices  \textit{qM}; \textit{dM};  \textit{aM}.}
	          \ForEach{two rows permutation $ \bm{q} $, $ \bm{d}  $ from $ n $ rows in $ P_i $ }
	{  
		\ForEach{$ q_j \in\bm{q},d_j \in\bm{d} $}    
		{ 		
				\lIf{$  q_j \leq d_j $} {$ q_j=0; d_j=0 $; } 
		}
	
		\uIf(\tcc*[f]{when Theorem \ref{thm:algo_extreme1} is true.}){$ \bm{q}$ and $ \bm{d}$ are  $ 0 $} 
	{	  qArr =$ [0]^n $;  dArr =$ [0]^n $;  aArr =$ [0]^n $; } 
	
	\uElse{

		Permutate $ \bm{q} $ and $ \bm{d} $ in the same way that makes $ \frac{q_1}{d_1} > \cdots > \frac{q_k}{d_k} $ where $ q_1,...,q_k $ are larger than $ 0 $;\;
		Let $qArr=\left[  q_1, ...,  \sum\nolimits_{i=1}^{n}q_i  \right] $;\;
		Let $dArr=\left[  d_1, ...,  \sum\nolimits_{i=1}^{n}d_i  \right] $;\;
		\ForEach{$i \in [1,n]$}    
		{ 		
			 aArr[i] = $ \log \frac{q_i - d_i}{qArr[i]*d_i - dArr[i]*q_i} $;\;
		}
	}
	Append  \textit{qArr, dArr, aArr} into new rows of \textit{qM, dM ,aM}, respectively;\;
	}
	\Return{ \textit{qM, dM, aM}}   
	\caption{Precompute the parameters.}
		\label{algo:precomp1}
	\end{scriptsize}
\end{algorithm}

We now use an example of $ \bm{q}=[0.2,0.3,0.5] $ and $ \bm{d}=[0.1,0,0.9] $ to demonstrate two notable points in Algorithm \ref{algo:precomp1}.
First, with such $ \bm{q}$ and $ \bm{d}$,  Line 7 results in  $ \bm{q}=[0.3,0.2,0] $ and  $ \bm{q}=[0,0.1,0] $ since only $ q_1=0.3>0 $ and  $ q_2=0.2>0 $.
Second, after the operation of Lines 10 and 11, we have \textit{aArr}= [Inf,1.47,NaN].

\begin{algorithm}[h]
	\begin{scriptsize}
		\SetKwRepeat{doWhile}{do}{while}
		\SetKwComment{tcc}{//}{}
		\DontPrintSemicolon 
		\KwIn{ $ \alpha $ (i.e., {\scriptsize ${ {{\alpha}^B_{t-1}}}$ or ${ {{\alpha}^F_{t+1}}}$}); $ qM $; $ dM $; $ aM $; $ \epsilon_t $ ({\scriptsize i.e., $\textit{PL}_0(\mathcal{M}^t) $}). }
		\KwOut{BPL or FPL at time $ t $}
		$ \mathcal{L} = 0$ ;\;
		\ForEach{ $ j \in [1,n(n-1)] $ }    
		{
			\textit{qArr, dArr, aArr} as $ j $-th rows in \textit{qM, dM ,aM}, respectively;\;
			Binary Search {\textit{aArr}[k]} that  {\textit{aArr}[k]} >$  \alpha \geq $ {\textit{aArr}[k+1]};\;
			PL = $\log  \frac{\textit{qArr}[k](e^{\alpha}-1)+1}{\textit{qArr}[k](e^{\alpha}-1)+1}$ ;\;
			\lIf{PL>$ \mathcal{L}  $}{$ \mathcal{L}  $ = PL;}
		}
		
		\Return{  {\scriptsize $ \mathcal{L}   + \epsilon_t$}}  
		\caption{Calculate BPL or FPL by  Precomputation.}
		\label{algo:calc_precomp}
	\end{scriptsize}
\end{algorithm}

Algorithm \ref{algo:precomp1} needs {\small $ O(n^3) $} time for precomputing the parameters  \textit{qM}, \textit{dM} and \textit{aM},   which only needs to be run one time and can be done offline.  
Algorithm \ref{algo:calc_precomp} for calculating privacy leakage at each time point needs {\small $ O(n^2\log n) $} time.

\subsection{Sub-linear Algorithm}
\label{sec:sub-linear}
In this section, we further design a sub-linear privacy leakage quantification algorithm by investigating how to generate a function of $ \mathcal{L}(\alpha) $, so that, given an arbitrary $ \alpha $, we can directly calculate the privacy loss.

\begin{corollary}
	\label{col:L_func}
Temporal privacy loss function $ \mathcal{L(\alpha)} $ can be represented as a piecewise function: $ \max_{\bm{q},\bm{d}\in P_i}  \log f_{\bm{q},\bm{d}}(\alpha)$.
\end{corollary}

\begin{figure}[t]
	\centering
	\vspace{-15pt}
	\includegraphics[scale=0.28]{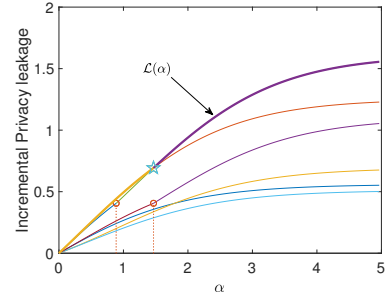} 
	\vspace{-10pt}
	\caption{ Illustration of function $ \mathcal{L(\alpha)} $ w.r.t. a $ 3 \times 3 $ transition matrix in Example \ref{example:L_func}. Each lines is $ f_{\bm{q},\bm{d}}(\alpha) $ w.r.t. distinct pairs of $ \bm{q}$ and $\bm{d} $. The bold line is $ \mathcal{L}(\alpha) $. The circle marks are transition points.}
	\label{fig:example_L_func}
	\vspace{-15pt}
\end{figure}

\begin{example}
	\label{example:L_func}
Figure \ref{fig:example_L_func} shows the $ \mathcal{L(\alpha)} $ function w.r.t.   {\scriptsize $ \big(\begin{smallmatrix} 0.1&0.2&0.7\\0.3&0.3&0.4\\ 0.5&0.3&0.2\end{smallmatrix} \big)$}.
Each line represents a piecewise function $ f_{\bm{q},\bm{d}}(\alpha) $ w.r.t. distinct pairs of $ \bm{q} $ and $  \bm{d}$  chosen from  $ P_i $.
Accoding to the definition of BPL and FPL,   $ \mathcal{L(\alpha)} $ function is   a piecewise function whose value is not less than any other functions for any $ \alpha $, e.g., the bold line in Figure \ref{fig:example_L_func}.
The pentagram, which is not any transition point, indicates an intersection of two piecewise functions.
\end{example}

Now, the challenge is how to find the ``top'' function which is larger than or equal to other piecewise functions.
A simple idea is to compare the piecewise functions in every sub-domains.
First, we list all transition points $ \alpha_1,...,\alpha_m $ of $ n(n-1) $ functions $ f_{\bm{q},\bm{d}}(\alpha) $  w.r.t. distinct pairs of $ \bm{q} $ and $ \bm{d} $ (i.e., all distinct values in \textit{aM}); then, we find the ``top'' piecewise function on each range between two consecutive transition points, which requires computation time $ O(n^3) $.
Despite the complexity, this  idea may not be correct because two piecewise functions may have an intersection between two consecutive transition points,  such as the pentagram  shown in Figure \ref{fig:example_L_func}.
Hence, we need a way to find such \textit{additional} transition points.
Our finding is that, if the top function is the same one at $ a_1 $ and $ a_2 $, then it is also the top function for any $ \alpha \in [a_1, a_2]$, which is formalized as the following theorem.

\begin{theorem}
\label{thm:gen_L_func}
Let {\small $ f(\alpha)=\frac{q(e^\alpha-1)+1}{d(e^\alpha-1)+1} $} and {\small $ f'(\alpha)=\frac{q'(e^\alpha-1)+1}{d'(e^\alpha-1)+1} $}.
If {\small $ f(a_1) \geq f'(a_1) $} and {\small $ f(a_2) \geq f'(a_2) $} in which {\small $ 0<a_1 < a_2 $}, we have {\small $ f(\alpha) \geq f'(\alpha) $} for any {\small $ a_1 \leq \alpha \leq a_2 $}.
\end{theorem}

Based on Theorem \ref{thm:gen_L_func}, we design Algorithm \ref{algo:gen_L_func} to generate a privacy loss function w.r.t. a given transition matrix.
Algorithm \ref{algo:gen_L_func} outputs a piecewise function of $ \mathcal{L(\alpha)}$:  $ qArr $, $ dArr $ are coefficients of the sub-functions, and $aArr $ contains the corresponding sub-domains.

We now analyze Algorithm \ref{algo:gen_L_func}, which  generate the privacy loss function $ \mathcal{L(\alpha)}$ in given a domain $ [a_1,a_m]$ by recursively find its sub-functions in sub-domains $ [a_1,a_k]$ and $ [a_k,a_m]$.
In Line 1, we check whether the parameters \textit{qM}, \textit{dM} and \textit{aM} are $ [0] $,  which implied that $ P_i $ is uniform and $ \mathcal{L}(\cdot) $ is 0 (see Lines 4 and 5 in Algorithm \ref{algo:precomp1}).
For convenience of the following analysis, we denote the definition of $ \mathcal{L}(\alpha) $  at point $\alpha = a_j $ by {\scriptsize $\mathcal{L}^j(\alpha) =\log \frac{q^j(e^{\alpha}-1)+1}{d^j(e^{\alpha}-1)+1}  $} with coefficients $ q^j $ and $ d^j $. 
In Lines 2, we obtain the definition of $ \mathcal{L}^1(\alpha) $  with coefficients $ q^1 $ and $ d^1 $.
In Line 3,  we obtain the definition of $ \mathcal{L}^m(\alpha) $ with coefficients $ q^m $ and $ d^m $.
From Line 4 to 20, we attempt to find the definition of $ \mathcal{L}(\alpha) $ at every point in $ [a_1, a_m] $.
There are two modules in this part.
The first one is from Line 5 to  12, in which we check that the definition of $ \mathcal{L}(\alpha) $ in $ (a_1, a_m) $ is $ \mathcal{L}^1(\alpha) $ or $ \mathcal{L}^m(\alpha) $.
The second module is from Line 13 to 19, in which we exam the definition of $ \mathcal{L}(\alpha) $ at point $ a_k $ (which may be the intersection of two sub-functions $ \mathcal{L}^1(\alpha) $ and $ \mathcal{L}^m(\alpha) $) and in the sub-domains $ [a_1, a_k] $ and $ [a_k, a_m] $.
Now, we dive into some details in these modules.
In Line 5, the condition is true when $ \mathcal{L}^1(\alpha) $ and $ \mathcal{L}^m(\alpha) $ are the same, or they intersect at $ \alpha = a_m $.
Then, $  \mathcal{L}^1  $ is the top function in $ [a_1, a_m] $ by Theorem \ref{thm:gen_L_func} because of {\scriptsize $  \mathcal{L}^1 (\alpha_1)\geq  \mathcal{L}^m(\alpha_1)$} and {\scriptsize $  \mathcal{L}^1 (\alpha_m) =  \mathcal{L}^m(\alpha_m) $}.
In Line 8,  the condition is true when  $ \mathcal{L}^1(\alpha) $ and $ \mathcal{L}^m(\alpha) $ intersect at $ a_1 $ or  they have no intersection in $ [a_1, a_m] $ (this implies {\scriptsize $ q^1=q^m $} and {\scriptsize $d^1 = d^m $}).
Then, $  \mathcal{L}^1  $ is the top function in $ [a_1, a_m] $ by Theorem \ref{thm:gen_L_func}.
In Lines 6$ \sim $8 and 10$ \sim $12, we store the coefficients and sub-domains in arrays.
From Line 14 to 19, we deal with the case of two functions intersecting at $ \alpha_k $ that {\scriptsize $ a_1 < \alpha_k < a_m $} by recursively invoking Algorithm \ref{algo:gen_L_func} to generate the sub-function of $ \mathcal{L}(\alpha) $ in {\scriptsize $ \left[\alpha_1, \alpha_k\right] $}  and {\scriptsize $ \left[\alpha_k, \alpha_m\right] $}.

\begin{algorithm}[h]
\begin{scriptsize}
\SetKwRepeat{doWhile}{do}{while}
\SetKwComment{tcc}{//}{}
\DontPrintSemicolon 
\KwIn{$ P_i $ (i.e., $  P_i^B$ or $  P_i^F$);  $ a_1 $; $ a_m $ ($ 0<a_1 \leq a_m$), \textit{qM}, \textit{dM}, \textit{aM}. }
\KwOut{vectors $ qArr $, $ dArr $ (i.e., coefficients), $  aArr $ (i.e., sub-domains).}
\lIf{ \textit{qM}, \textit{dM}, \textit{aM} are $ [0] $}{\Return \textit{qArr,dArr,aArr} as [0]}
Find  {\textit{maxPL1}} , $ {q^1} $,  $ {d^1} $ using \textbf{Algorithm \ref{algo:precomp1}} with  $ a_1 $, \textit{qM, dM, aM}, $ \epsilon_t=0 $;\;
	Find {\textit{maxPLm}} , $ {q^m} $,  $ {d^m} $ using \textbf{Algorithm \ref{algo:precomp1}} with $ a_m $, \textit{qM, dM, aM}, $ \epsilon_t=0 $; \;
$ k= \frac{q^m+d^1-q^1-d^m}{q^1*d^m-q^m*d^1} $;\;
\uIf{ $ a_1=a_m $ {\bf or} 
	{\scriptsize  {maxPLm} = $\log \frac{q^1(e^{a_m}-1)+1}{d^1(e^{a_m}-1)+1}$}
}{ 
$ aArr= \left[a_m\right]$; \tcc*{initialize vectors.}
$ qArr= \left[q^1\right] $; \; 
$ dArr= \left[d^1\right]$; \; }
\uElseIf{  
	 {\scriptsize {maxPL1} $  = \log \frac{q^m(e^{a_1}-1)+1}{d^m(e^{a_1}-1)+1}$ }  {\bf or}
 	k<=0			}
{
	$ aArr= \left[a_m\right]$; \tcc*{initialize vectors.}
	$ qArr= \left[q^m\right] $; \; 
	$ dArr= \left[d^m\right]$; \; }
\uElse{
$ a_k=\log(k +1)$; \;
Find  { $ qArr_k $, $ dArr_k $, $  aArr_k $} using \textbf{Algorithm \ref{algo:gen_L_func}} with  $ P_i, a_1, a_k $;\;
Find   $ qArr_m $, $ dArr_m $, $  aArr_m $ using \textbf{Algorithm \ref{algo:gen_L_func}} with $ P_i, a_k, a_m $;\;
$ aArr = \left[aArr_k, aArr_m \right]$; \tcc*{concatenate two vectors.}
$ qArr = \left[qArr_k, qArr_m \right] $;\;
$ dArr= \left[dArr_k, dArr_m \right] $;\;
}
\Return{ $ qArr $, $ dArr $, $  aArr $.}  
\caption{Generate Privacy Loss Function $ \mathcal{L}(\alpha) $.}
\label{algo:gen_L_func}
\end{scriptsize}
\end{algorithm}

When obtaining function $ \mathcal{L}(\alpha) $, we can directly calculate BPL or FPL as shown in Algorithm \ref{algo:cal_bpl_func}.
In Line 1 of Algorithm \ref{algo:cal_bpl_func}, we can perform a binary search because $  aArr $ is sorted.

\begin{algorithm}[h]
	\begin{scriptsize}
		\SetKwRepeat{doWhile}{do}{while}
		\SetKwComment{tcc}{//}{}
		\DontPrintSemicolon 
		\KwIn{ $ \alpha $ (i.e., {\scriptsize ${ {{\alpha}^B_{t-1}}}$ or ${ {{\alpha}^F_{t+1}}}$}); $ qArr $, $ dArr $, $  aArr $ (i.e., $ \mathcal{L}(\alpha) $); $ \epsilon_t $. }
		\KwOut{BPL or FPL at time $ t $}
		\lIf{ \textit{qArr}, \textit{dArr}, \textit{aArr} are $ [0] $}{ \Return $ \epsilon_t $}
		Binary Search $ \alpha_k $ in $  aArr $ that $ \alpha_k \geq \alpha $ and $ \alpha_{k+1}< \alpha$;\;
		\Return{ {\scriptsize $  \log \frac{qk(e^\alpha-1)+1}{dk(e^\alpha-1)+1} + \epsilon_t$}}  
		\caption{Calculate BPL or FPL by  $ \mathcal{L(\alpha)} $.}
		\label{algo:cal_bpl_func}
	\end{scriptsize}
\end{algorithm}

\textbf{Complexity}.
The  complexity of Algorithm \ref{algo:cal_bpl_func} for calculating privacy leakage at one time point is $ O(\log n) $, which makes it very efficient even for large $ n $ and $ T $.
 Algorithm \ref{algo:gen_L_func} itself requires $ O(n^2\log n+m\log m) $ time where $ m $ is the amount of transition points in $ [a_1, a_m] $, and its  parameters \textit{qM}, \textit{dM} and \textit{aM} need to be calculated by Algorithm \ref{algo:precomp1} which requires $ O(n^3) $ time; hence, in total,  generating $ \mathcal{L}(\cdot) $ needs $ O(n^3+m\log m) $ times.

\section{Bounding Temporal Privacy Leakage}
\label{sec:release_mechanism}
In this section, we design two privacy budget allocation strategies that can be used to convert a traditional DP mechanism into one protecting against TPL.

We first investigate the upper bound of BPL and FPL.
We have demonstrated that  BPL and FPL may increase over time as shown in Figure \ref{fig:tpl}.
A natural question is that: is there a limit of  BPL and FPL over time.

\begin{theorem}
        \label{thm:sup}
        Given a transition matrix {\footnotesize $ P_i^B $}  (or {\footnotesize $ P_i^F $}),
        let $ q $ and $ d $ be the ones that give the maximum value in Equation \eqref{eq:lfp_sol_b}  and {\small $ q \neq d $}.
        For {\small $ \mathcal{M}^t $} that satisfies $ \epsilon $-DP at each {\small $ t \in [1,T] $},
        there are four cases regarding the supremum of BPL (or FPL) over time.
        
\begin{myAlignSSS}
        \begin{cases}
                {\tiny \log \frac{{\sqrt { 4d{e^\epsilon }(1-q ) + {{( d + q{e^\epsilon } -1)}^2} }  + d + q{e^\epsilon } - 1}}{{2d}} }&   d \neq 0  \\
                \log \frac{(1-q) e^\epsilon}{1-q e^\epsilon }     &  {\footnotesize d=0 \text{ and }  q \neq 1 \text{ and }  \epsilon \leq \log(1/q)} \\
                \inf  &  {\footnotesize d=0 \text{ and }   q \neq 1 \text{ and }  \epsilon > \log(1/q)}  \\
                \inf &  d=0  \text{ and }  q=1 
        \end{cases} \nonumber
\end{myAlignSSS}
         
\end{theorem}




\setlength{\textfloatsep}{0pt}
\begin{figure}[t]
	\vspace{-15pt}
        \centering
        \hspace*{-5pt}
        \includegraphics[scale=0.3]{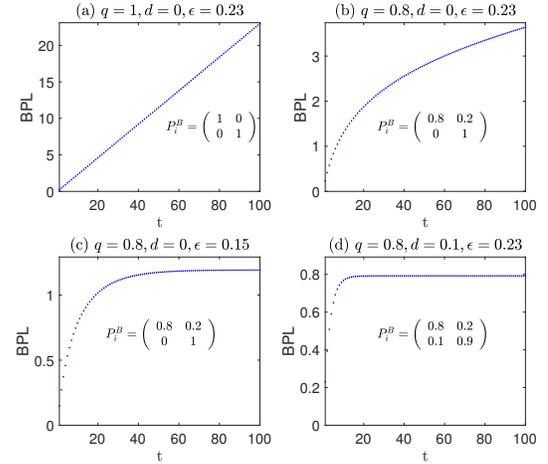} 
        \vspace{-15pt}
        \caption{Examples of the maximum BPL over time.}
        \label{fig:example_sup}
\end{figure}

\begin{example}[The supremum of  BPL over time]
        Suppose that $ \mathcal{M}^t $ satisfies $ \epsilon $-DP at each time point.
        In Figure \ref{fig:example_sup}, it shows the maximum BPL w.r.t. different  $ \epsilon $ and different transition matrices of {\small $ P_i^B $}.
        In (a) and (b), the supremum does not exist.
        In (c) and (d), we can calculate the supremum using Theorem \ref{thm:sup}.
        The results are in line with the ones from computing BPL step by step at each time point using Algorithm \ref{algo:cal_bpl}.
\end{example}

\begin{algorithm}[h]
	\begin{scriptsize}
		\SetKwRepeat{doWhile}{do}{while}
		\SetKwComment{tcc}{//}{}
		\DontPrintSemicolon 
		\KwIn{ $ \epsilon_t $ ; \textit{qM}; \textit{dM}. }
		\KwOut{the supremum of BPL or FPL over time, $ q $, $ d $}
		$ sup =0$; $ q =0$; $ d =0$ \;
		\ForEach{{$ q_c $} $ \in $ \textit{qM}, {$ d_c $} $ \in $\textit{dM}}
	{		Calculate a candidate $ sup_c $ using Theorem \ref{thm:sup} ;\;
			\lIf{ $ sup $ < $ sup_c $}{  $ sup $ = $ sup_c $; $ q $ = $ q_c $; $ d $ = $ d_c $}
	}
	
		\Return{$  sup $, $ q $, $ d $}  
		\caption{Find Supermum of BPL or FPL.}
		\label{algo:findSup}
	\end{scriptsize}
\end{algorithm}

Algorithm \ref{algo:findSup} for finding supreme of BPL or FPL is correct because, according to Theorem \ref{thm:lfp_k_case}, all possible $ q $ and $ d $ that give the maximum value in Equation \eqref{eq:lfp_sol_b}  are in the matrices $ qM $ and $ dM $.
Algorithm \ref{algo:findSup}  is useful not only for designing privacy budget allocation strategies in this section, but also for setting an appropriate parameter $ a_m $ as the input of Algorithm \ref{algo:gen_L_func} because we will show in experiments that a larger $ a_m $ makes  Algorithm \ref{algo:gen_L_func} time-consuming (while, too small $ a_m $ may result in failing  to calculate privacy leakage from $ \mathcal{L}(\alpha) $ if the input $ \alpha $ may be larger than $ a_m $).

\textbf{Achieving $ \alpha $-DP$ _\mathcal{T} $ by limiting upper bound.}
We now design a privacy budget allocation strategy utilizing Theorem \ref{thm:sup} to bound TPL.
Theorem \ref{thm:sup} tells us that, 
if it is not the strongest temporal correlation (i.e., {\small $ d=0 $} and {\small $ q=1 $}), 
we may bound  BPL or FPL within a desired value by allocating an appropriate privacy budget to a traditional DP mechanism at each time point.
In other words, we want a constant $ \epsilon_t $ that guarantee the supremum of TPL, which is equal to the sum of the supremum of BPL and the supremum of BPL subtracting  $ \epsilon_t $ by Equation \eqref{eq:tpl_comp2}, will not larger than $ \alpha $.
Based on this idea, we design Algorithm \ref{algo:release1} for solving such $ \epsilon_t $.

\vspace{-10pt}
\begin{algorithm}[h]
        \begin{scriptsize}
                \SetKwRepeat{doWhile}{do}{while}
                \SetKwComment{tcc}{//}{}
                \DontPrintSemicolon 
                \KwIn{ $  P_i^B$ and $  P_i^F$;  $ \alpha $ (desired privacy level). }
                \KwOut{Privacy budgets $ \epsilon_t$ satisfying $ \alpha $-DP$_\mathcal{T} $ at each  $ t $}
                	Find {$ qM_B $, $ dM_B $} using \textbf{Algorithm} \ref{algo:gen_L_func} with $ P_i^B $;\;
                	Find {$ qM_F$, $ dM_F $} using \textbf{Algorithm} \ref{algo:gen_L_func} with $ P_i^F $;\;
                	$ bingo = false $; $range= \alpha $;  $ e = 0.5*\alpha$;\;
                	\doWhile(\tcc*[f]{binary search.}){$ bingo = false $}
                	{
                		$ range = 0.5*range $;\;
                		Calculate $ sup_B $ by \textbf{Algorithm} \ref{algo:findSup} with $ e $ and $ qM_B $, $ dM_B $;\;
                		Calculate $ sup_F $ by \textbf{Algorithm} \ref{algo:findSup} with $ e $ and $ qM_F $, $ dM_F $;\;
                		
                		\lIf{$sup_B + sup_F - e =  \alpha$}{$ bingo = true $}
                		\lElseIf{$ sup_B + sup_F - e >  \alpha$}{ $ e = e- range $}
                		\lElse{$ e = e + range $}             		
                		
                	}
        
            
                \Return{ $ \epsilon_t = e$ }  
                \caption{{\small Achieving $ \alpha $-DP$ _\mathcal{T} $ by upper bound}}
                \label{algo:release1}
        \end{scriptsize}
\end{algorithm}
\vspace{-10pt}

\textbf{Achieving $ \alpha $-DP$ _\mathcal{T} $ by privacy leakage quantification.}
Algorithm \ref{algo:release1} allocates privacy budgets in a conservative way:
when $ T $ is short, the privacy leakage may not be increased to the supremum.
We now design Algorithm \ref{algo:release2} to overcome this drawback.
Observing the supremum of backward privacy loss in Figure \ref{fig:example_sup}(c)(d), BPL at the  first time point is much less than the supremum.
Similarly, it is easy to see that FPL at the last time point is much less than its supremum.
Hence, we attempt to allocate more privacy budgets to {\scriptsize $ \mathcal{M}^1 $} and {\scriptsize $ \mathcal{M}^T $} so that the temporal privacy leakage at every time points are \textit{exactly} equal to the desired level.
Specifically, if we want that BPL at two consecutive time points are exactly the same value {\footnotesize $ \alpha^B $}, i.e., {\scriptsize $  BPL(\mathcal{M}^t)= BPL(\mathcal{M}^{t+1}) = \alpha^B$}, we can derive that {\scriptsize $ \epsilon_t = \epsilon_{t+1} $} for {\scriptsize $ t\geq 2 $} (it is true for {\scriptsize $ t=1 $ } only when {\scriptsize $ \mathcal{L}^B(\cdot)=0 $}). 
Applying the same logic to FPL, we have a new strategy for allocating privacy budgets: assigning larger privacy budgets at time points $ 1 $ and $ T $, and constant values at time $ [2,T] $ to make sure that BPL at time points $ [1, T-1] $ are the same, denoted by $ \alpha^B $, and FPL at time $ [2, T] $ are the same, denoted by $ \alpha^F$.
Hence, we have {\scriptsize $ \epsilon_1  =\alpha^B$} and {\scriptsize $ \epsilon_T = \alpha^F $}.
Let  the value of privacy budget at $ [2,T] $ be $  \epsilon_m$.
We have (i) {\scriptsize $ \mathcal{L}^B(\alpha^B) + \epsilon_m = \alpha^B$}, (ii) {\scriptsize $ \mathcal{L}^F(\alpha^F) + \epsilon_m = \alpha^F$} and (iii)  {\scriptsize $ \alpha^B + \alpha^F - \epsilon_m = \alpha$}.
Combing (i) and (iii), we have Equation \eqref{eq:release2_a};  Combing (ii) and (iii), we have Equation \eqref{eq:release2_b}.
\begin{myAlignSS}
\mathcal{L}^B(\alpha^B) + \alpha^F = \alpha \label{eq:release2_a}\\
\mathcal{L}^F(\alpha^F) + \alpha^B = \alpha \label{eq:release2_b}
\end{myAlignSS}

\vspace{-8pt}
Based on the above idea, we design Algorithm \ref{algo:release2} to solve $ \alpha^B $  and $ \alpha^F $.
Since $ \alpha^B $ should be in $ [0, \alpha] $, we heuristically initialize $ \alpha^B $ with $ 0.5*\alpha $ in Line 3 and then use binary search to find appropriate $ \alpha^B $ and $ \alpha^F $ that satisfy Equations \eqref{eq:release2_a} and \eqref{eq:release2_b}.

\begin{algorithm}[h]
	\begin{scriptsize}
		\SetKwRepeat{doWhile}{do}{while}
		\SetKwComment{tcc}{//}{}
		\DontPrintSemicolon 
		\KwIn{ $  P_i^B$ and $  P_i^F$ ;  $ \alpha $ (desired privacy level for user $ i $). }
		\KwOut{Privacy budgets $ \epsilon_t, t\in [1,T] $ satisfying $ \alpha $-DP$_\mathcal{T} $ at each  $ t $}
			Find \textit{$ qM^B $, $ dM^B $, $ aM^B $} using \textbf{Algorithm} \ref{algo:gen_L_func} with  $  P_i^B$;\;
			Find \textit{$ qM^F $, $ dM^F$, $ aM^F $} using \textbf{Algorithm} \ref{algo:gen_L_func} with  $  P_i^F$;\;
			$ bingo = false $; $range= \alpha $;  $ a^B = 0.5* \alpha$;\;
			\doWhile(\tcc*[f]{binary search.}){$ bingo = false $}
			{
				$ range = 0.5*range $ ;\;
				Find $ L^B $ by \textbf{Algorithm} \ref{algo:cal_bpl_func} with $ a^B $, \textit{$ qM^B $, $ dM^B $, $ aM^B $}, $ \epsilon=0$;\;
				$ a^F = \alpha - L^B $; \tcc*{Equation \eqref{eq:release2_a}.}
				Find $ L^F $ by \textbf{Algorithm} \ref{algo:cal_bpl_func} with $ a^F$ with  \textit{$ qM^F $, $ dM^F$, $ aM^F $}, $ \epsilon =0$;\;
				\lIf{$  L^F +a^B = \alpha$}{$ bingo = true $}
				\lElseIf(\tcc*[f]{Equation \eqref{eq:release2_b}.}){$ L^F +a^B < \alpha$}{ $ a^B = a^B + range $}
				\lElse{$ a^B = a^B - range $}             		
				
			}
			
		\Return{ $ \epsilon_1 = a^B $; $ \epsilon_{t}= a^B+a^F - \alpha, t\in[2,T-1] $;  $ \epsilon_T = a^F $}  
		\caption{{\small Achieving $ \alpha $-DP$ _\mathcal{T} $ by quantification}}
		\label{algo:release2}
	\end{scriptsize}
\end{algorithm}

\vspace{-15pt}
\section{Experimental Evaluation}
\label{sec:exp}
In this section, we design experiments for the following:
(1) verifying the runtime and correctness of our privacy leakage quantification algorithms,
(2) investigating the impact of the temporal correlations on privacy leakage and
(3) evaluating the data release Algorithms \ref{algo:release1} and \ref{algo:release2}.
We implemented all the algorithms\footnote{Souce code: https://github.com/brahms2013/TPL} in Matlab2017b and conducted the experiments on a machine with an Intel Core i7 2.6 GHz CPU and 16G RAM  running macOS High Sierra.

%

\vspace{-10pt}
\subsection{Runtime of Privacy Quantification Algorithms}
\label{subsec:runtime}
In this section, we compare the runtime of our algorithms  with  \textsf{IBM ILOG CPLEX}\footnote{http://www-01.ibm.com/software/commerce/optimization/cplex-optimizer/. We use version 12.7.1.}, which is a well-known software for solving optimization problems, e.g., the linear-fractional programming problem \eqref{eq:lfp}$ \sim $\eqref{eq:lfp_end2} in our setting.

For verifying the correctness of  three privacy  quantifying algorithms, Algorithm \ref{algo:cal_bpl}, Algorithm \ref{algo:calc_precomp} and Algorithm \ref{algo:cal_bpl_func}, we generate $ 100 $ random transition matrices with dimension size $ n=30 $ and comparing the calculation results with the one solving LFP problem using \textsf{CPLEX}.
We verified that  all results obtained from our algorithms are identical to the one using \textsf{CPLEX}  w.r.t. the same transition matrix.

For testing the runtime of our algorithms, we run them $ 30 $ times with randomly generated transition matrices, and run \textsf{CPLEX} one time (because it is very time-consuming),
and then calculate the average runtime for each of them.
Since Algorithm \ref{algo:calc_precomp} needs  parameters that are precomputed by Algorithm \ref{algo:precomp1}, and Algorithm \ref{algo:cal_bpl_func} needs   $ \mathcal{L}(\cdot) $ that can be obtained using  Algorithm \ref{algo:gen_L_func}, we also test the runtime of these precomputations.
The results are shown in Figure \ref{fig:ex_runtime}.

\begin{figure*}[t]
        \centering
        \includegraphics[scale=0.38]{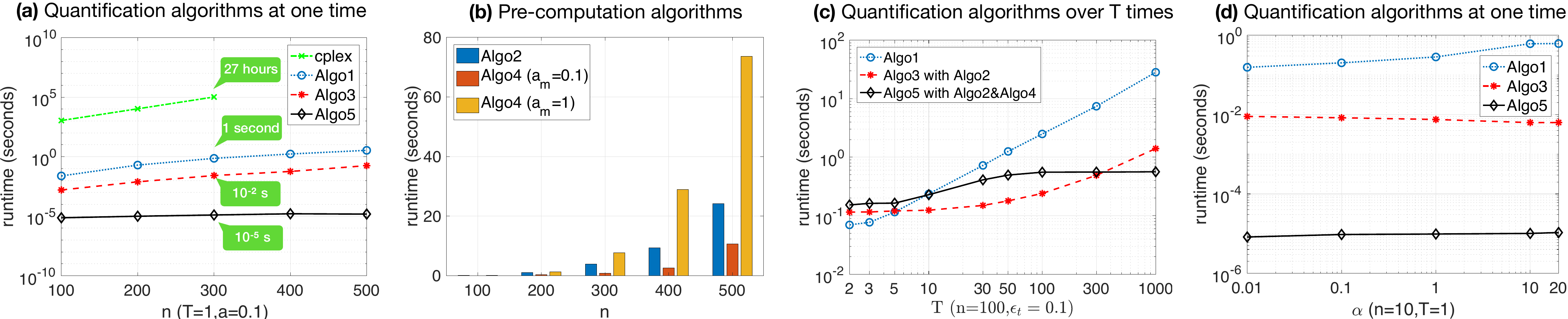} 
        \vspace{-10pt}
        \caption{Runtime of Temporal Privacy Leakage Quantification Algorithms.}
        \label{fig:ex_runtime}
        \vspace{-15pt}
\end{figure*}

\textbf{Runtime vs. $ n $}.
In Figures \ref{fig:ex_runtime}(a) and (b), we show the runtime of  privacy quantification algorithms and precomputation algorithms, respectively,
In Figure \ref{fig:ex_runtime}(a), each algorithm takes inputs of {\normalsize $ \alpha=0.1 $} and {\normalsize $ n \times n $} random probability matrices.
The runtime of all algorithms increase along with $ n $ because  the number of variables in our LFP problem is $ n $.
The proposed Algorithms \ref{algo:cal_bpl},  \ref{algo:calc_precomp} and \ref{algo:cal_bpl_func} significantly outperform \textsf{CPLEX}.
In Figure \ref{fig:ex_runtime}(b), we test precomputation procedures.
We observed that all algorithms are increasing with $ n $, but Algorithm \ref{algo:gen_L_func}  is more susceptible to $ a_m $.
Algorithm  \ref{algo:gen_L_func} with a larger $ a_m $ results in higher runtime because it performs binary search in $ [0, a_m] $.

\cyrev{
	These results are in line with our complexity analysis, which shows the complexity of Algorithms \ref{algo:precomp1} and  \ref{algo:gen_L_func}  are {\scriptsize $ O(n^3) $} and {\scriptsize $ O(n^2\log n + m\log m) $} ($ m $ is the amount of transition points and increasing with $ a_m $),  respectively.
	However, when $ n $ or $ a_m $ is large,  pre-computation Algorithms \ref{algo:precomp1} and  \ref{algo:gen_L_func} are time-consuming because we need to find optimal solutions for $ n*(n-1) $ LFP problems given a  $ n $-dimension transition matrix.
	This can be improved by advanced  computation tools such as parallel computing because here each LFP problem is independently solved, and such computation only needs to be run one time before starting to release private data. 
	Another interesting way to improve the runtime is to find the relationship between the optimal solutions of different LFP problems given a transition matrix (so that we can prune some computations).
	We defer this to future study.
}


\textbf{Runtime vs. $ T $}.
In Figure \ref{fig:ex_runtime}(c), we test the runtime of each privacy quantification algorithm integrated with their precomputations over different length of time points.
We want to know how can we benefit from these precomputations over time.
All algorithms take inputs of $ 100\times 100 $ matrices and {\scriptsize $ \epsilon_t=0.1$} for each time point $ t $.  
The parameters  of Algorithm \ref{algo:gen_L_func} need to be initialized by Algorithm \ref{algo:precomp1}, so we take them as an integrated module along with Algorithm \ref{algo:cal_bpl_func}.
It shows that,  Algorithm \ref{algo:cal_bpl} runs fast  if $ T $ is small. 
Algorithm \ref{algo:calc_precomp} becomes the most preferable if $ T $ is in $ [5,300] $.
However, when $ T $ is larger than 300, Algorithm \ref{algo:cal_bpl_func} with its precomputation (Algorithm \ref{algo:gen_L_func} with {\footnotesize $ a_m=(T+1)*\epsilon_t $} which is the worst case of supremum) is the fast one and its runtime is almost constant with the increase of $ T $.
Therefore, there is no best algorithm in efficiency without a known $ T $, but we can choose appropriate algorithm adaptively.

\textbf{Runtime vs. $ \alpha $}.
In Figure \ref{fig:ex_runtime}(d), we show that, a larger previous BPL (or the next FPL), i.e., $ \alpha $, may lead to higher runtime of Algorithm \ref{algo:cal_bpl}, whereas other algorithms are relatively stable for varying $ \alpha $.
The reason is that, when $ \alpha $ is large, 
Algorithm \ref{algo:cal_bpl} may take more time in 
Lines $ 9 $ and $ 10 $ for updating  each pair of $ q_j  \in \bm{q}^+$ and $ d_j  \in \bm{d}^+$ to satisfy Inequality \eqref{eq:cond1}.
An update in Line $ 10 $ is more likely to occur due to a large $ \alpha $ because {\footnotesize $  \frac{q(e^{{\alpha}} -1) + 1 }{d(e^{{\alpha}} - 1) + 1} $} is increasing with $ \alpha $.
However, such growth of runtime along with $ \alpha $ will not last so long because the update happens $ n $ times in the worse case.
As shown in Figure \ref{fig:ex_runtime}(b), when $ \alpha >10 $, the runtime of Algorithm \ref{algo:cal_bpl} becomes stable.

\vspace{-5pt}
\subsection{Impact of Temporal Correlations on TPL}
\label{subsec:ex_impact}

In this section, for the ease of exposition, we only present the impact of temporal correlations on BPL because
the growth of BPL and FPL are in the same way.

\textbf{The setting of temporal correlations.}
To evaluate if our privacy loss quantification algorithms can perform well under different degrees of temporal correlations, we find a way to {generate} the transition matrices to eliminate the effect of different correlation estimation algorithms or datasets.
First, we generate a transition matrix indicating the ``strongest'' correlation that contains probability $ 1.0 $ in its diagonal cells (this type of transition matrix will lead to an upper bound of TPL).
Then, we perform \textit{Laplacian smoothing}\cite{sorkine_laplacian_2004} 
 to \textit{uniformize} the probabilities of $ P_i $ (the uniform transition matrix will lead to an low bound of TPL). 
Let $ {p_{jk}}  $ be an element at the $ j $th row and $ k $th column of the matrix $ P_i $.
The uniformized probabilities $ \hat{p}_{jk}  $ are generated using Equation \eqref{eq:laplace_smoothing}, where $ s $ is a positive parameter that controls the degrees of uniformity of probabilities in each row.
Hence, a smaller $s$ means stronger temporal correlation.
We note that, different $ s $ are only comparable under the same $ n $. 
\begin{myAlignSSS}
	\hat{p}_{jk} = \frac{ p_{jk} +s }{\sum_{u=1}^{n}{(p_{ju} +s)}} \label{eq:laplace_smoothing}
\end{myAlignSSS}


We examined  $ s $ values ranging from 0.005 to 1 and set $n $ to $ 50 $ and $ 200 $. 
Let $ \varepsilon $ be the privacy budget of $ \mathcal{M}^t $ at each time point. 
We test $ \varepsilon =1$ and 0.1. 
The results are shown in Figure \ref{fig:ex_tpl} and  are summarized as follows.

\textbf{Privacy Leakage vs. $ \bm{s} $.}
Figure \ref{fig:ex_tpl} shows that the privacy leakage caused by a non-trivial temporal correlation will increase over time, and such growth first increases sharply and then remains stable because it is gradually close to its supremum.
The increase caused by a stronger temporal correlations (i.e., smaller $ s $) is steeper, and the time for the increase is longer.
Consequently, stronger correlations result in higher privacy leakage.

\textbf{Privacy Leakage vs. $ \bm{\varepsilon} $.}
Comparing Figures \ref{fig:ex_tpl}(a) and (b), we found that  $0.1$-DP  significantly delayed the growth of privacy leakage.
Taking $ s=0.005 $, for example, the noticeable increase continues for almost 8 timestamps when $ \varepsilon=1 $ (Figures \ref{fig:ex_tpl}(a)), whereas it continues for approximately 80 timestamps when $ \varepsilon=0.1 $ (Figures \ref{fig:ex_tpl}(b)).
However, after a sufficient long time, the privacy leakage in the case of $ \varepsilon=0.1 $ is not substantially lower than that of $ \varepsilon=1 $ under stronger temporal correlations.  
This is because, although the privacy leakage is eliminated at each time point by setting a small privacy budget, the adversaries can eventually learn sufficient information from the continuous releases.

\textbf{Privacy Leakage vs. $ n$.}
Under the same $ s $, TPL is smaller when $ n$ (dimension of the transition matrix) is larger, as shown in the lines $ s=0.005 $ with $ n =50$ and $ n =200$ of Figure \ref{fig:ex_tpl}. 
This is because the transition matrices tend to be uniform (weaker correlations) when the dimension is larger.


\begin{figure}[t]
        \centering
        \includegraphics[width=0.90\linewidth]{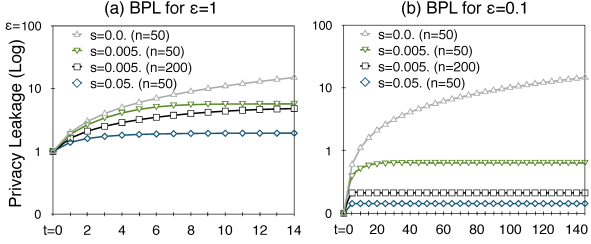} 
        \vspace{-10pt}
        \caption{Evaluation of BPL under different degrees of  correlations.}
        \label{fig:ex_tpl}
\end{figure}

\begin{figure}[t]
        \hspace*{-15pt}\includegraphics[scale=0.28]{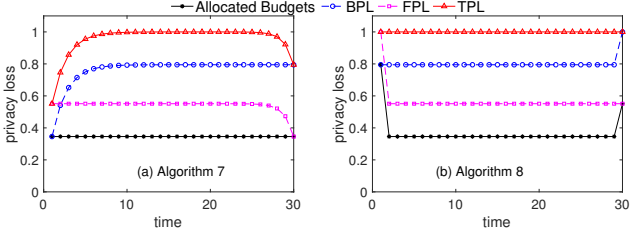} 
        \vspace{-22pt}
        \caption{Privacy Budget Allocation Schemes for  {\small $ 1 $-DP$_\mathcal{T} $}.}
        \label{fig:ex_budgets}
\end{figure}

\subsection{Evaluation of Data Releasing Algorithms}
In this section, we first show a visualization of privacy allocation of  Algorithms \ref{algo:release1} and \ref{algo:release2}, then we compare the data utility in terms of Laplace noise.

Figure \ref{fig:ex_budgets} shows an example of budget allocation, w.r.t. {\scriptsize $ P_i^B=\big(\begin{smallmatrix} 0.8&0.2\\ 0.2&0.8\end{smallmatrix} \big)$} and {\scriptsize $ P_i^F=\big(\begin{smallmatrix} 0.8&0.2\\ 0.1&0.9\end{smallmatrix} \big)$}.
The goal is $ 1 $-DP$ _\mathcal{T} $.
It is easy to see that Algorithm \ref{algo:release2} has better data utility because it exactly achieves the desired privacy level.

Figure \ref{fig:ex_utility} shows the data utility of Algorithms \ref{algo:release1} and \ref{algo:release2} with $ 2 $-DP$ _\mathcal{T} $.
We calculate the absolute value of the Laplace noise with the allocated budgets (as shown in Figure \ref{fig:ex_budgets}).
Higher value of noise indicates lower data utility.
In Figure \ref{fig:ex_utility}(a), we test the data utility under backward and forward temporal correlation both with parameter $ s=0.001 $, which means relatively strong correlation.
It shows that, when $ T $ is short, Algorithm \ref{algo:release2} outperforms Algorithm \ref{algo:release1}.
In Figure \ref{fig:ex_utility}(b), we investigate the data utility under different degree of correlations.
The dash line indicates the absolute value of Laplace noise if no temporal correlation exists (privacy budget is $ 2 $).
It is easy to see that the  data utility significantly decays under strong correlation $ s=0.01 $.

\begin{figure}[t]
        \centering
        \includegraphics[scale=0.26]{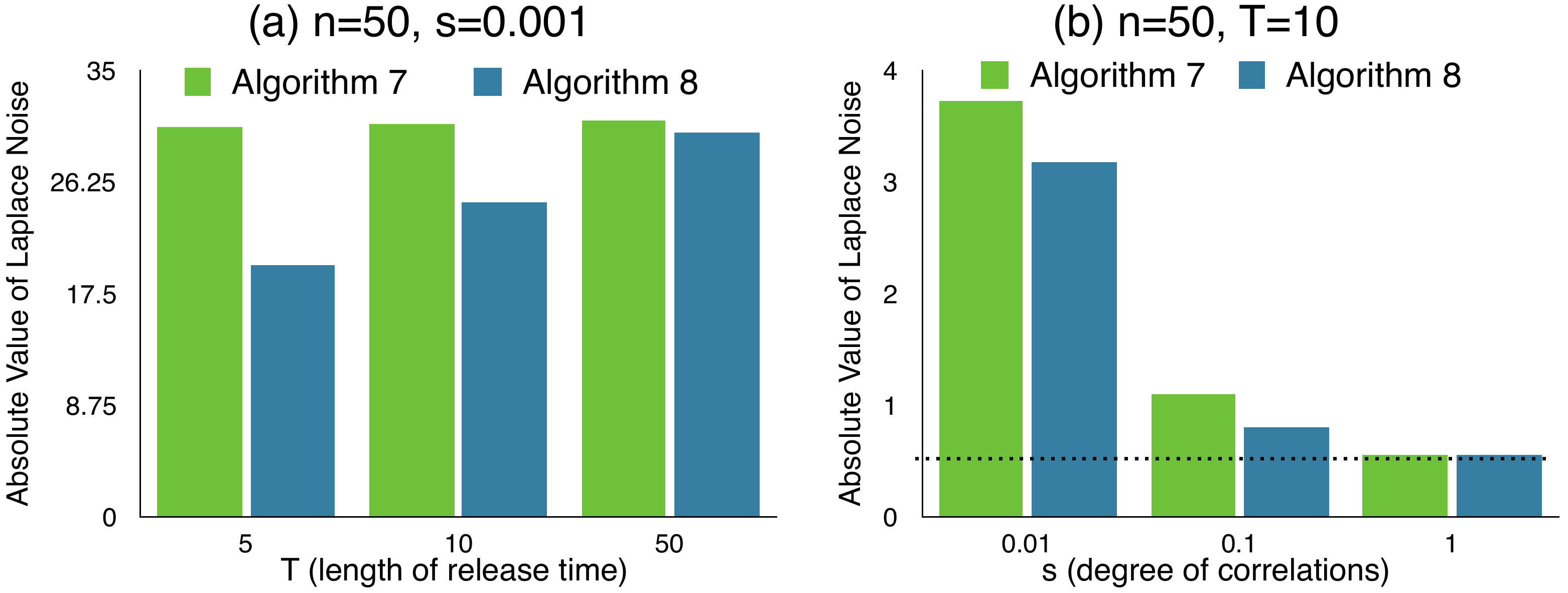} 
        \vspace{-10pt}
        \caption{Data utility of $ 2 $-DP$ _\mathcal{T} $ mechanisms.}
        \label{fig:ex_utility}
\end{figure}

\vspace{-1pt}
\section{Related Work}

Dwork et al. first studied \textit{differential privacy under continual observation} and proposed event-level/user-level privacy\cite{dwork_differential_2010-1}\cite{dwork_differential_2010}.
A plethora of studies have been conducted to investigate different problems in this setting, such as 
high dimensional data\cite{acs_case_2014}
\cite{xiao_dpcube:_2014}, 
infinite sequence\cite{kellaris_differentially_2014}\cite{cao_differentially_2015}\cite{cao_differentially_2016},
and real-time publishing\cite{fan_fast:_2013}
\cite{li_differentially_2015}. 
To the best of our knowledge, no study has reported the risk of differential privacy under temporal correlations for the continuous aggregate release setting.
Although \cite{xiao_protecting_2015} have considered a similar adversarial model in which the adversaries have prior knowledge of temporal correlations represented by Markov chains, they proposed  a mechanism extending differential privacy for releasing a private location, whereas we  focus on the scenario of continuous aggregate release with DP.

Several studies have questioned whether differential privacy is valid for correlated data. 
Kifer and Machanavajjhala\cite{kifer_no_2011} \cite{kifer_rigorous_2012}\cite{kifer_pufferfish:_2014} 
first raised the important issue that differential privacy may not guarantee privacy if  adversaries know the data correlations between tuples.
They\cite{kifer_no_2011} argued that it is not possible to ensure any utility in addition to privacy without making assumptions about the data-generating distribution and the background knowledge available to an adversary.
To this end, they proposed a general and customizable privacy framework called \textit{PufferFish}\cite{kifer_rigorous_2012}\cite{kifer_pufferfish:_2014}, in which the potential secrets, discriminative pairs, and data generation need to be explicitly defined. 
Song et al. \cite{song_pufferfish_2017} proposed Markov Quilt Mechanism when the correlations can be modeled by Bayesian Network.
Yang et al.\cite{yang_bayesian_2015}  investigated differential privacy on correlated tuples described using a proposed Gaussian correlation model.
The privacy leakage w.r.t. adversaries with specified prior knowledge can be efficiently computed.
Zhu et al. \cite{zhu_correlated_2015} proposed correlated differential privacy by defining the sensitivity of queries on correlated data.
Liu et al. \cite{changchang_liu_dependence_2016} proposed dependent differential privacy by introducing dependence coefficients for analyzing the sensitivity of different queries under probabilistic dependence between tuples.
Most of the above works dead with correlations between users in the database, i.e., user-user correlations, in the setting of one-shot data release, whereas we deal with the correlation among single user's data at different time points, i.e., temporal correlations, and focusing on the dynamic change of privacy guarantee in continuous data release.

On the other hand,  it is still controversial \cite{li_differential_2016} what should be the guarantee of DP on correlated data.
Li et al. \cite{li_differential_2016} proposed \textit{Personal Data Principle}  for clarifying the privacy guarantee of DP on correlated data.
It states that an individual's privacy is not violated if no data about the individual is used.
By doing this, one can ignore any correlation between this individual's data and other users' data, i.e.,user-user correlations.
On the other hand, the question of ``what is individual's data'', or ``what should be an appropriate notion of negiboring databases'' is tricky in many application scenarios such as genomic data.  
If we apply Personal Data Principle to the setting of continuous data release,  event-level privacy is not a good fit for protecting individual's privacy because a user's data at each time point is only a part of his/her whole data in the streaming database.
Our work shares the same insight with Personal Data Principle on this point: we show that the privacy loss of event-level privacy may increase over time under temporal correlation, while the guarantee of user-privacy is as expected.
We note that, when the data stream is infinite or the end of data release is unknown, we can only resort to event-level privacy or w-event privacy \cite{kellaris_differentially_2014}.
Our work provides useful tools against TPL in such setting.


\vspace{-10pt}
\section{Conclusions}
In this paper, we quantified the risk of differential privacy under temporal  correlations by formalizing, analyzing and calculating the privacy loss against adversaries who have knowledge of temporal  correlations.
Our analysis shows the privacy loss of event-level privacy may increase over time, while the privacy guarantee of user-level privacy is as expected.
We design fast algorithms for quantifying temporal privacy leakage which enables  private data release in real-time.
This work opens up interesting future research directions, such as investigating the privacy leakage under temporal correlations combining with other type of correlation models, and use our methods to enhance previous studies that neglected the effect of temporal  correlations in continuous data release.

\ifCLASSOPTIONcaptionsoff
  \newpage
\fi



%

\vspace{-1pt}
\bibliographystyle{IEEEtran}
\bibliography{./TKDE18-v5}

%

\vspace{-25pt}
\begin{IEEEbiography}[{\includegraphics[width=1in,height=1.25in,clip,keepaspectratio]{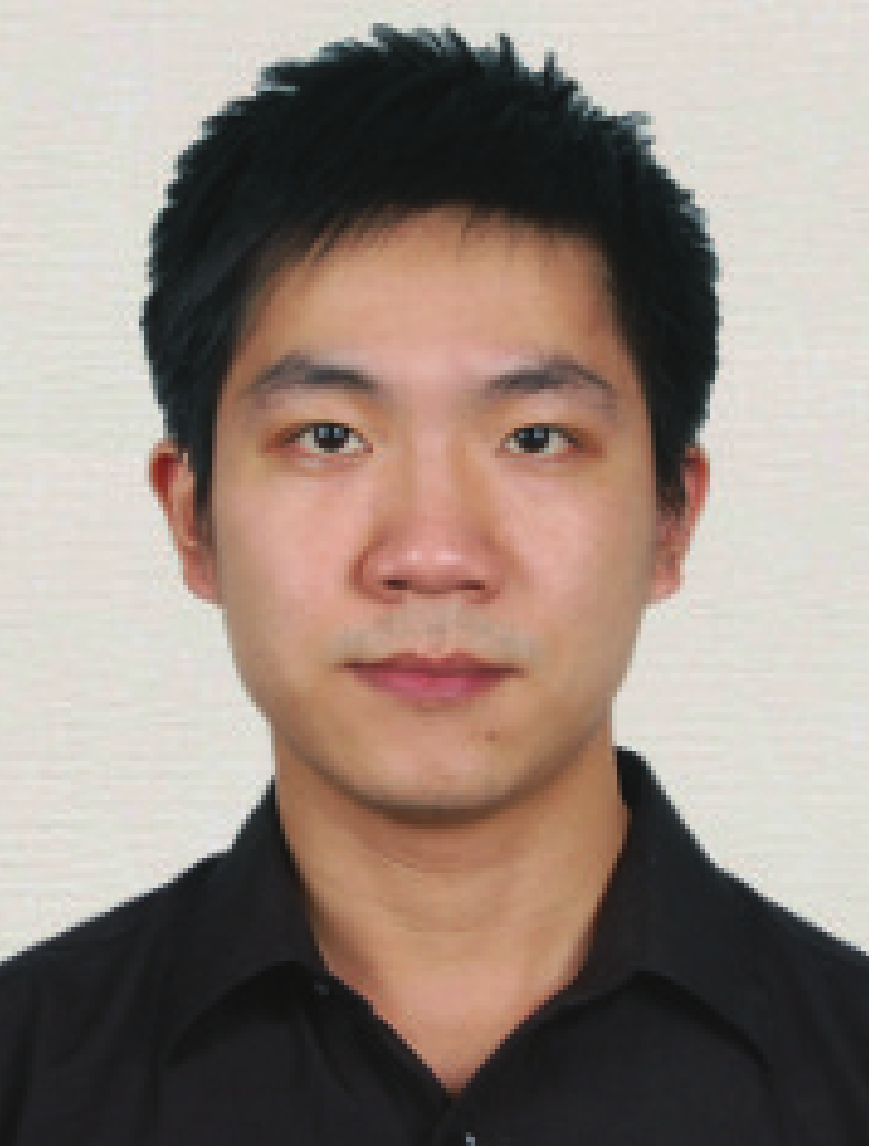}}]{Yang Cao}
  received B.S. degree from the School of Software Engineering, Northwestern Polytechnical University, China, in 2008, the M.S. degree and Ph.D. degree from the Graduate School of Informatics, Kyoto University, Japan, in 2014 and 2017, respectively. He is currently a Postdoctoral Fellow at the department of Math and Computer Science, Emory University, USA. His research interests are privacy preserving data publishing and mining.
\end{IEEEbiography}


\vspace{-25pt}
\begin{IEEEbiography}[{\includegraphics[width=1in,height=1.25in,clip,keepaspectratio]{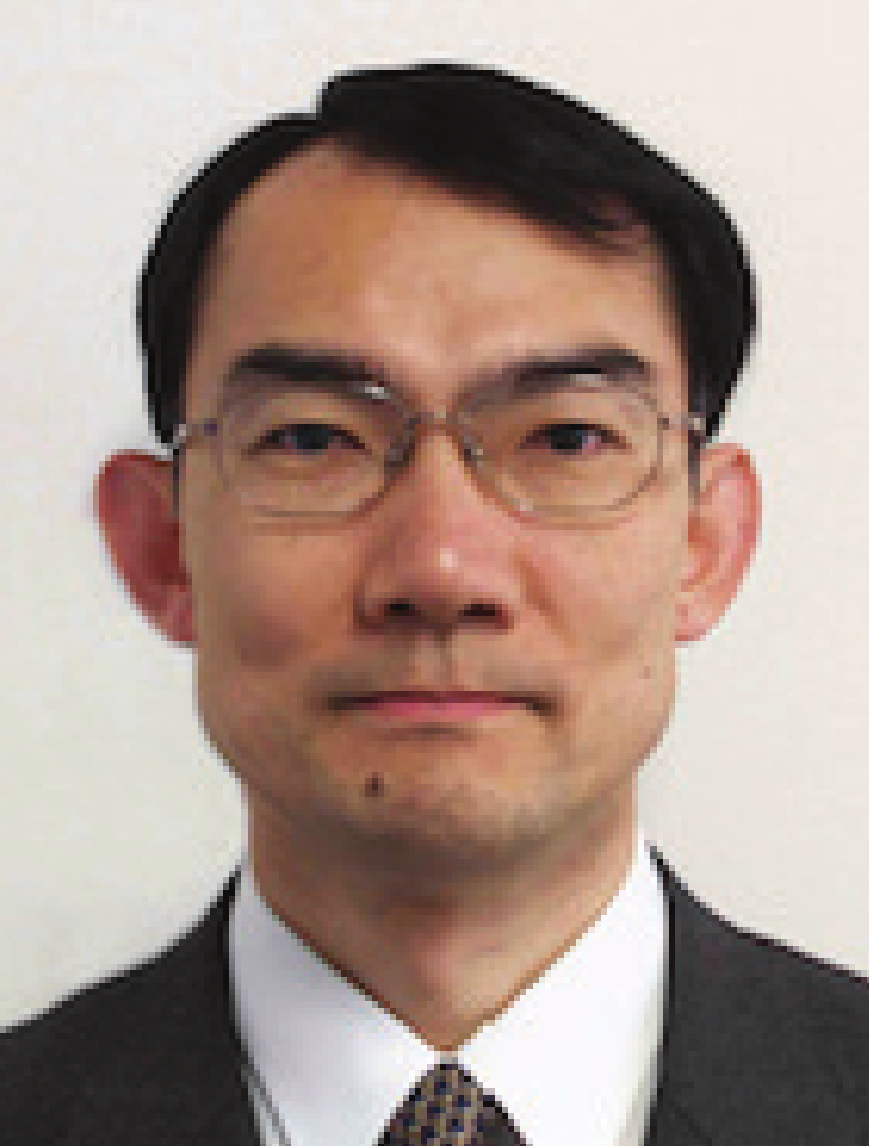}}]{Masatoshi Yoshikawa}
received the B.E., M.E. and Ph.D. degrees from Department of Information Science, Kyoto University in 1980, 1982 and 1985, respectively. Before joining Graduate School of Informatics, Kyoto University as a
professor in 2006, he has been a faculty member of Kyoto Sangyo University,
Nara Institute of Science and Technology, and Nagoya University.
His general research interests are in the area of databases.
His current research interests include multi-user routing algorithms and
services, theory and practice of privacy protection, and medical data
mining. He is a member of ACM and IPSJ.
\end{IEEEbiography}

\vspace{-25pt}
\begin{IEEEbiography}[{\includegraphics[width=1in,height=1.25in,clip]{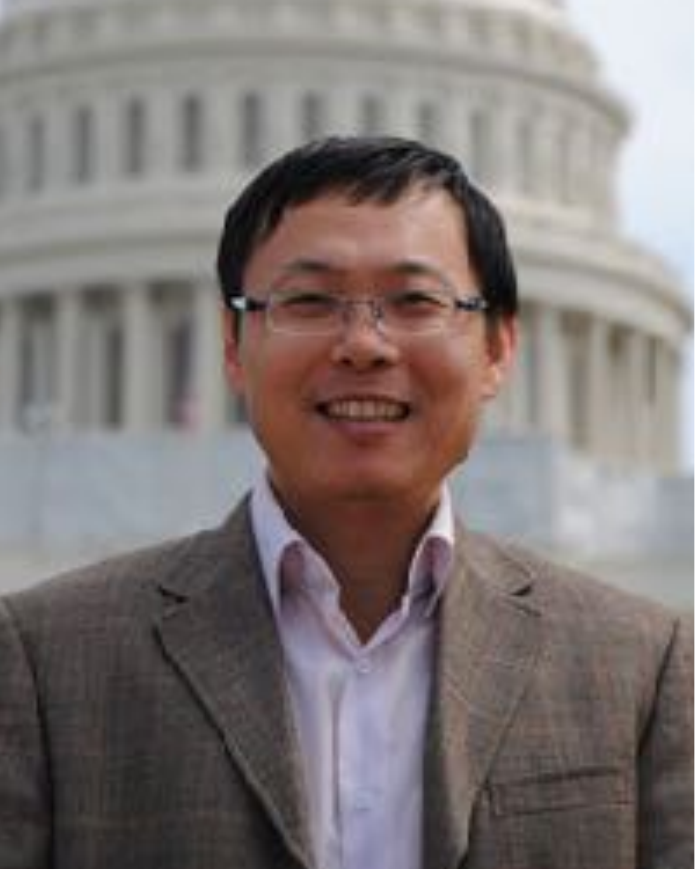}}]{Yonghui Xiao} received 3 B.S. degrees from Xi'an Jiaotong University, China, in 2005. He obtained the M.S. degree from Tsinghua University in 2011 after spending two years working in industry. He obtained the Ph.D. degree in the department of Math and Computer Science at Emory University in 2017. He is currently a software engineer at Google. 
\end{IEEEbiography}

\vspace{-25pt}
\begin{IEEEbiography}[{\includegraphics[width=1in,height=1.25in,clip,keepaspectratio]{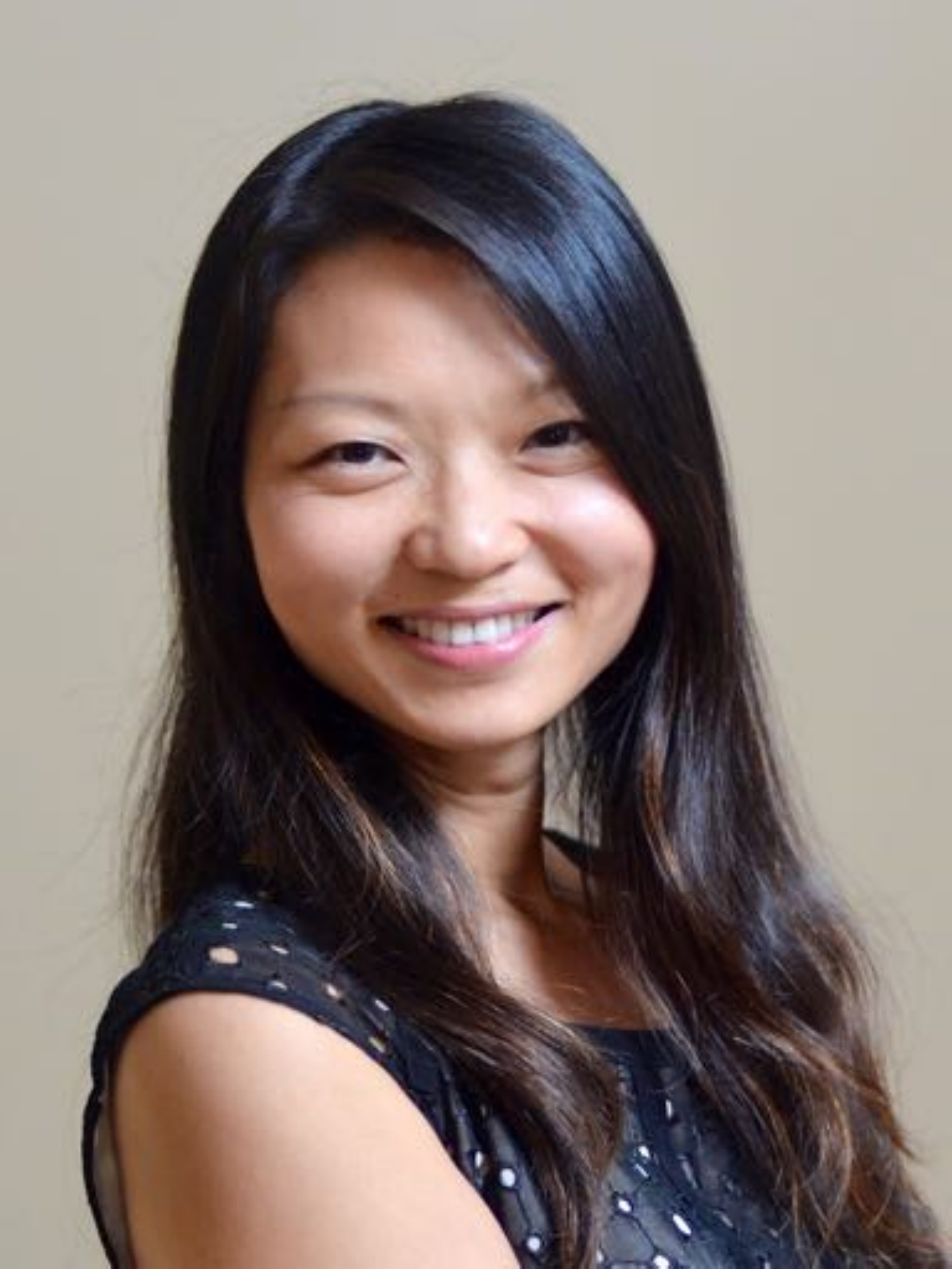}}]{Li Xiong}
	is a Winship Distinguished Research Professor of Computer Science (and Biomedical Informatics) at Emory University. She holds a PhD from Georgia Institute of Technology, an MS from Johns Hopkins University, and a BS from University of Science and Technology of China, all in Computer Science. She and her research group, Assured Information Management and Sharing (AIMS), conduct research that addresses both fundamental and applied questions at the interface of data privacy and security, spatiotemporal data management, and health informatics.
\end{IEEEbiography}




\newpage
\newpage
\newpage

\appendices
\section{Proof of Theorem \ref{thm:analysis_tpl}}
\cyrev{ 
\label{pf:analysis_tpl}
We need to prove
\begin{myAlignSSS}
	&  {\textit{TPL}}(A_i^\mathcal{T} ,\mathcal{M}^t)= \label{pf:eq:tpl1} \\
	&\sup_{ l^t, {l_i^t}' } 
	\left( 
	\sup_{\bm{r}^1}
	\log \frac{\Pr(\bm{r}^1|{l_i^t,D_\mathcal{K}^t})}{\Pr({\bm{r}^1|{{l_i^t}',D_\mathcal{K}^t}})}
	+
	\dots 
	+
	\sup_{\bm{r}^T}
	\log \frac{\Pr(\bm{r}^T|{l_i^t,D_\mathcal{K}^t})}{\Pr({\bm{r}^T|{{l_i^t}',D_\mathcal{K}^t}})} 
	\right)
	\label{pf:eq:tpl_expanded0}  \\
	&=
	\sup_{ l^t, {l_i^t}', \bm{r}^1}
	\log \frac{\Pr(\bm{r}^1|{l_i^t,D_\mathcal{K}^t})}{\Pr({\bm{r}^1|{{l_i^t}',D_\mathcal{K}^t}})}
	+
	\dots 
	+
	\sup_{l^t, {l_i^t}', \bm{r}^T}
	\log \frac{\Pr(\bm{r}^T|{l_i^t,D_\mathcal{K}^t})}{\Pr({\bm{r}^T|{{l_i^t}',D_\mathcal{K}^t}})} 
	\label{pf:eq:tpl_expanded1} \\
	&=
			\sup_{\substack{\bm{r}^1,...,\bm{r}^t,\\{l_i^t},{l_i^t}'}} 
			\log \frac{\Pr(\bm{r}^1,...,\bm{r}^t|{l_i^t,D_\mathcal{K}^t})}{\Pr({\bm{r}^1,...,\bm{r}^t|{{l_i^t}',D_\mathcal{K}^t}})}
		+
			\sup_{\substack{\bm{r}^t,...,\bm{r}^T,\\ {l_i^t},{l_i^t}'}} 
			\log \frac{\Pr(\bm{r}^t,...,\bm{r}^T|{l_i^t,D_\mathcal{K}^t})}{\Pr(\bm{r}^t,...,{\bm{r}^T|{{l_i^t}',D_\mathcal{K}^t}})} \nonumber\\
		&\hspace{2em} -
			\sup_{\bm{r}^t,{l_i^t},{l_i^t}'} 
			\log \frac{\Pr(\bm{r}^t|{l_i^t,D_\mathcal{K}^t})}{\Pr({\bm{r}^t|{{l_i^t}',D_\mathcal{K}^t}})} \label{pf:eq:tpl_expanded2}
\end{myAlignSSS}
Because  $\bm{r}^t $ at different $ t\in[1,T] $ are independent given $D^t=\{l_i^t,D_\mathcal{K}^t  \}$ or ${D^t}'=\{{l_i^t}',D_\mathcal{K}^t  \}$, we can derive Equation \eqref{pf:eq:tpl_expanded0}  from Equation \eqref{pf:eq:tpl1}, and derive Equation \eqref{pf:eq:tpl_expanded2} from Equation \eqref{pf:eq:tpl_expanded1}.
The remaining question is how to prove the derivation of  Equation \eqref{pf:eq:tpl_expanded1} from  Equation \eqref{pf:eq:tpl_expanded0}.
It is true because $ \mathcal{M}^t $ is the same  DP mechanism at different $ t $, and the output domains $ Range(\mathcal{M}^k)=Range(\mathcal{M}^j)$ for any $ k,j \in [1,T] $.
We explain in details in below.
\\
Assume we have $ \{ \bar{l_i^t}, \bar{l_i^{t'}} \} $ maximizing the first item in Equation \eqref{pf:eq:tpl_expanded1} with a certain value of  $ \bm{r}^1$, i.e.,
\begin{myAlignSSS} \{ \bar{l_i^t}, \bar{l_i^{t'}} \} = \argmax_{\bar{l_i^t}, \bar{l_i^{t'}}}  \log \frac{\Pr(\bm{r}^1|{l_i^t,D_\mathcal{K}^t})}{\Pr({\bm{r}^1|{{l_i^t}',D_\mathcal{K}^t}})}. \nonumber
\end{myAlignSSS}
Then,  $ \{ \bar{l_i^t}, \bar{l_i^{t'}} \} $ also maximizes other items in Equation \eqref{pf:eq:tpl_expanded1} such as {\tiny $ \sup\limits_{\substack{\bm{r}^k, {l_i^t},{l_i^t}'}} 
\log \frac{\Pr(\bm{r}^k|{l_i^t,D_\mathcal{K}^t})}{\Pr(\bm{r}^k|{{l_i^t}',D_\mathcal{K}^t})} $} with a certain value of  $ \bm{r}^k$ where $ k\in[1,T] $.
\noindent
\\
Formally, that is to say, there exists $ \bar{\bm{r}^k} \in Range(\mathcal{M}^k)$, for any $ k\in[1,T] $, that we have
\begin{myAlignSSS}   \log \frac{\Pr( \bar{\bm{r}^k}|{\bar{l_i^t},D_\mathcal{K}^t})}{\Pr({ \bar{\bm{r}^k}|{\bar{{l_i^t}'},D_\mathcal{K}^t}})}
	= 
\sup_{l^t, {l_i^t}', \bm{r}^k}
\log \frac{\Pr(\bm{r}^k|{l_i^t,D_\mathcal{K}^t})}{\Pr({\bm{r}^k|{{l_i^t}',D_\mathcal{K}^t}})}. \nonumber
\end{myAlignSSS}
Hence, we can derive Equation \eqref{pf:eq:tpl_expanded1} from Equation \eqref{pf:eq:tpl_expanded0}.
The theorem follows.\footnote{We note that, if the DP mechanisms at each time are significantly different, or the sensitivity of queries at each time are different, Equation \eqref{pf:eq:tpl_expanded1} may not be the supremum of TPL, but an upper bound of TPL.}
}
 
\section{Proof of Theorem \ref{thm:composition}}
\label{pf:composition}

\begin{proof}
        According to Definition \ref{def:tpl} of TPL,
        the overall privacy leakage of $\mathcal{M}^t $ and $\mathcal{M}^{t+1} $ in terms of DP$ _\mathcal{T} $ are
        \vspace{5pt}
        \begin{myAlignSSS}
                &\sup_{ D^t, D_\mathcal{M}^t,\ldots, {D_\mathcal{M}^{t+j}}', \bm{r}^1, \ldots, \bm{r}^T } \log
                {{\Pr ({\bm{r}^1}, \ldots, {\bm{r}^{T}}|{D^t,D^{t+1}})} \over {\Pr ({\bm{r}^{1}}, \ldots ,{\bm{r}^{T}}|{D^t}',{D^{t+1}}')}} \nonumber
                \\
                &= \sup_{\substack{ D^t, {D^t}', \\ \bm{r}^1, \ldots, \bm{r}^t} } \log
                {{\Pr ({\bm{r}^1}, \ldots, {\bm{r}^{t}}|{D^t})} \over {\Pr ({\bm{r}^{1}}, \ldots ,{\bm{r}^{t}}|{D^t}')}} 
                + \sup_{\substack{ D^t, {D^t}', \\ \bm{r}^{t+1}, \ldots, \bm{r}^T} } \log
                {{\Pr ({\bm{r}^{t+1}}, \ldots, {\bm{r}^{T}}|{D^{t+1}})} \over {\Pr ({\bm{r}^{t+1}}, \ldots ,{\bm{r}^{T}}|{D^{t+1}}')}} 
                \nonumber \\
                &= \textit{BPL}(\mathcal{M}^t) +  \textit{FPL}(\mathcal{M}^{t+1}) \nonumber \\
                &=\alpha_{t}^B + \alpha_{t+j}^F. \nonumber
        \end{myAlignSSS}
        
        Therefore, the theorem follows when $ j=1 $.
        
        According to Definition \ref{def:combined_pl},
        the overall privacy leakage of {\small $\mathcal{M}^t ,\ldots,\mathcal{M}^{t+j} $} in terms of DP$ _\mathcal{T} $ is as follows.
        \vspace{5pt}
        \begin{myAlignSSS}
                &\sup_{\substack{ D^t,...,D^{t+j},\\ {D^t}',...,{D^{t+j}}',\\ \bm{r}^1, \ldots, \bm{r}^T} } \log
                {{\Pr ({\bm{r}^1}, \ldots, {\bm{r}^{T}}|{D^t,\ldots,D^{t+j}})} \over {\Pr ({\bm{r}^{1}}, \ldots ,{\bm{r}^{T}}|{D^t}',\ldots, {D^{t+j}}')}} \nonumber
                \\
                &= \sup_{\substack{ D^t, {D^t}', \\ \bm{r}^1, \ldots, \bm{r}^t} } \log
                {{\Pr ({\bm{r}^1}, \ldots, {\bm{r}^{t}}|{D^t})} \over {\Pr ({\bm{r}^{1}}, \ldots ,{\bm{r}^{t}}|{D^t}')}} 
                + \sup_{\substack{ D^{t+j}, {D^{t+j}}', \\ \bm{r}^{t+j}, \ldots, \bm{r}^T} } \log
                {{\Pr ({\bm{r}^{t+j}}, \ldots, {\bm{r}^{T}}|{D^{t+j}})} \over {\Pr ({\bm{r}^{t+1}}, \ldots ,{\bm{r}^{T}}|{D^{t+j}}')}} 
                \nonumber \\
                &+\sup_{\substack{ D^{t+1},\\ {D^{t+1}}',\\ {\bm{r}}^{t+1} }} \log
                {{\Pr ({\bm{r}^{t+1}}|{D^{t+1}})} \over {\Pr ({\bm{r}^{t+1}}|{D^{t+1}}')}} 
                + \cdots +
                \sup_{\substack{ D^{t+j-1},\\ {D^{t+j-1}}',\\ {\bm{r}}^{t+j-1} }} \log
                {{\Pr ({\bm{r}^{t+j-1}}|{D^{t+j-1}})} \over {\Pr ({\bm{r}^{t+j-1}}|{D^{t+j-1}}')}} 
                \nonumber \\
                &= \textit{BPL}(\mathcal{M}^t) +  \textit{FPL}(\mathcal{M}^{t+j}) 
                +\textit{PL}_0(\mathcal{M}^{t+1}) + \cdots \textit{PL}_0(\mathcal{M}^{t+j-1}) \nonumber \\
                &=\alpha_{t}^B + \alpha_{t+j}^F + \sum_{k=1}^{k=j-1}\epsilon_{t+k}. \nonumber
        \end{myAlignSSS}
        
        Therefore, the theorem follows when $ j\geq 2 $.
\end{proof}

\section{Proof of Theorem \ref{thm:lfp}}
\label{appex:thm_tpl}
We need Dinkelbach's Theorem and Lemma \ref{lem:lfp_proof}  in our proof.
\begin{theorem}[Dinkelbach's Theorem\cite{dinkelbach_nonlinear_1967}]
        \label{thm:Dinkelbach}
        In a linear-fractional programming problem, suppose that the variable vector is $\boldsymbol{x}$ and the objective function is represented as {\small $\frac{Q(\boldsymbol{x})}{D(\boldsymbol{x})}$}.
        Vector $\boldsymbol{x^*}$ is an optimal solution if and only if 
        \begin{myAlignS}
                \max\{Q(\boldsymbol{x})-\lambda*D(\boldsymbol{x})\} = 0 
                \text{ where } \lambda= \frac{Q(\boldsymbol{x^*})}{D(\boldsymbol{x^*})}.  \label{eq:dink}
        \end{myAlignS}
\end{theorem}

\begin{lemma}
        \label{lem:lfp_proof}
        For the following maximization problem ({\small $ k_1,...,k_n\in \mathbb{R}$}) with the same constraints as the ones in the linear-fractional programming \eqref{eq:lfp}$ \sim $\eqref{eq:lfp_end2},
        \begin{myAlignS}
                \textnormal{maximize } &  \text{ } k_1 x_1+\cdots+k_{n} x_{n}  \nonumber  \\
                \textnormal{subject to  } 
                & \text{ } e^{-{{\alpha}^B_{t-1}}}*{x_k}  \leq {x_j}\leq e^{{{\alpha}^B_{t-1}}}*{x_k} ,   \nonumber \\
                &\text{ }  0 < {x_j}< 1 \text{ and }  0<x_k<1, \nonumber  \\
                & \text{ where } x_j,x_k \in \bm{x}, \text{ } j,k \in  [1,n]. \nonumber
        \end{myAlignS} 
        
        \noindent
        an optimal solution is as follows: if  {\small $ k_i>0 $}, let {\small $ x_i=e^{{{\alpha}^B_{t-1}}} m$} where $ m$ is a positive real number; if  {\small $ k_i\leq 0 $}, let {\small $  x_i=m$}.
\end{lemma}

\begin{proof}
        Without loss of generality, we suppose that the smallest value in the optimal solution is $ x_{n}  $.
        Let $ y_j $ be {\footnotesize $\frac{ x_j}{x_{n} }$} for {\footnotesize $ j\in[1,n-1] $}; then, {\footnotesize $ 1 \leq y_j \leq {e^{{\alpha}^B_{t-1}}} $}.
        Replacing $ x_j $ with $ y_j $ and setting {\footnotesize $ x_{n} =m$}, we have a new objective function {\footnotesize $ \frac{1}{m} * ( k_1 y_1+\cdots+k_{n-1} y_{n-1}+k_{n} ) $} whose solution is equivalent to the original one.
        Because the only constraint is {\footnotesize $ 1 \leq y_j \leq e^{{{\alpha}^B_{t-1}}} $}, the following is an optimal solution for the maximum objective function: if {\footnotesize $ k_j>0 $}, let {\footnotesize $ y_j =e^{{{\alpha}^B_{t-1}}} $}; if {\footnotesize $ k_j \leq 0 $}, let {\small $ y_j=1$}. 
\end{proof}

\begin{proof}[Proof of Theorem \ref{thm:lfp}]
        We first prove that, under the conditions shown in Theorem \ref{thm:lfp}, i.e., Inequalities \eqref{eq:cond1} and \eqref{eq:cond2}, an optimal solution of the problem \eqref{eq:lfp}$ \sim $\eqref{eq:lfp_end2} is:
        
        \begin{myAlignS}
                \bm{x}^* =
                \begin{cases}
                        x_j = e^{{{\alpha}^B_{t-1}}}*m & x_j \in \bm{x}^+ \\
                        x_k = m & x_k \in \bm{x}^- \label{eq:opt_sol}
                \end{cases},
        \end{myAlignS}
where $ m $ is a positive real number.

        In combination with Dinkelbach's Theorem, we rewrite our objective function as {\footnotesize $\frac{Q(\bm{x})}{D(\bm{x})}  $} in which {\footnotesize $ Q(\bm{x})=\bm{q}\bm{x} $} and  {\footnotesize $ D(\bm{x})=\bm{d}\bm{x} $}.
        Substituting {\footnotesize $ \bm{x}^* $} of Equation \eqref{eq:opt_sol} into {\footnotesize $ Q(\bm{x}) $} and {\footnotesize $ D(\bm{x}) $}, we have {\footnotesize $ Q(\bm{x}^*)=q(e^{{{\alpha}^B_{t-1}}}-1)+1 $} and {\footnotesize $ D(\bm{x}^*)=d(e^{{{\alpha}^B_{t-1}}}-1)+1 $} (recall that {\footnotesize $ q=\sum{\bm{q}^+} $} and {\footnotesize $d=\sum{\bm{d}^+} $}.
        Then, we can rewrite Inequalities \eqref{eq:cond1} and \eqref{eq:cond2} in Theorem \ref{thm:lfp} as follows.
        \begin{myAlignS}
                &\frac{q_j}{d_j} > \frac{Q(\bm{x}^*)}{D(\bm{x}^*)}, & \forall j\in[1,n] \text{ where } q_j \in \bm{q}^+, d_j \in \bm{d}^+ \label{eq:1}\\
                &\frac{q_k}{d_k} \leq \frac{Q(\bm{x}^*) }{D(\bm{x}^*)}, & \forall k\in[1,n] \text{ where } q_k \in \bm{q}^-, d_k \in \bm{d}^- \label{eq:2}
        \end{myAlignS}

        Because {\footnotesize $ D(\bm{x}^*)>0 $}, to prove $ \bm{x}^* $ in \eqref{eq:opt_sol} is an optimal solution for {\footnotesize $ \frac{Q(\bm{x})}{D(\bm{x})} $}, we only need to prove the maximum value of the following equation is equal to $ 0 $. 
        \begin{myAlignS}
                \text{maximize } \{D(\bm{x}^*)  Q(\bm{x}) - Q(\bm{x}^*)  D(\bm{x}) \}=0.  \label{eq:thm_max}
        \end{myAlignS}
        
        We expand the above equation as follows.
        \begin{myAlignS}
                &\text{Eqn.}\eqref{eq:thm_max}=D(\bm{x}^*)  (\bm{q}^+ \bm{x}^+ + \bm{q}^- \bm{x}^-) 
                - 
                Q(\bm{x}^*)  (\bm{d}^+ \bm{x}^+ + \bm{d}^- \bm{x}^-)  \nonumber \\
                &= \big(D(\bm{x}^*)  \bm{q}^+ - Q(\bm{x}^*)  \bm{d}^+\big) \bm{x}^+
                + 
                \big(D(\bm{x}^*)  \bm{q}^- - Q(\bm{x}^*)   \bm{d}^-\big) \bm{x}^-   \label{eq:last}  
        \end{myAlignS}

        By Equations \eqref{eq:1} and \eqref{eq:2}), we have
        {\footnotesize $ D(\bm{x}^*)  \bm{q}^+ - Q(\bm{x}^*)  \bm{d}^+ >0$} and {\footnotesize $ D(\bm{x}^*)  \bm{q}^- - Q(\bm{x}^*)   \bm{d}^- \leq 0 $}.
        Hence, according to Lemma \ref{lem:lfp_proof}, we can obtain the maximum value in Equation \eqref{eq:thm_max} by setting {\footnotesize $ \bm{x}^+=[e^{{\alpha}^B_{t-1}}*m] $} and {\footnotesize $ \bm{x}^-=[m] $} where {\footnotesize $ m $} is a positive real number.
        We obtain the maximum value in Equation \eqref{eq:thm_max}.
        \begin{myAlignSS}
                &\big((D(\bm{x}^*) q - Q(\bm{x}^*) d)  e^{\varepsilon} m
                + 
                \big(D(\bm{x}^*)  (1-q) - Q(\bm{x}^*)  (1-d)) \big) m \nonumber \\
                =& \big(D(\bm{x}^*) (q e^\varepsilon+(1-q)) - Q(\bm{x}^*) (d  e^\varepsilon +(1-d))\big)m  \nonumber \\
                =& \big(D(\bm{x}^*) Q(\bm{x}^*) - Q(\bm{x}^*) D(\bm{x}^*)\big) m = 0  \nonumber 
        \end{myAlignSS}

        \noindent
        Therefore, by Dinkelbach's Theorem, $ \bm{x}^* $ is an optimal solution for the problem \eqref{eq:lfp}$ \sim $\eqref{eq:lfp_end2}.
        Substituting them into the objective function \eqref{eq:lfp}, we obtain the maximum value 
        {\footnotesize $ \frac{q(e^{{\alpha}_{t-1}^B} -1) + 1 }{d(e^{{\alpha}_{t-1}^B} - 1) + 1} $}.
\end{proof}

%
\section{Proof of Corollary \ref{col:q_d}}
\label{appex:col_q_d}
\begin{proof}
        The proof is by contradiction.
        \\ \indent
        First, assume that {\footnotesize $ q=d $}, i.e., {\footnotesize $ \sum{\bm{q}^+} = \sum{\bm{d}^+}$}.
        That is, there exist {\footnotesize $ q_j \in  \bm{q}^+$}  and {\footnotesize $ d_j \in  \bm{d}^+$} satisfying {\footnotesize $ \frac{q_j}{d_j} \leq 1$}.
        This leads to a contradiction to Inequality \eqref{eq:cond1} because of  {\footnotesize $ \frac{q_i}{d_i} > \frac{q(e^{{\alpha}_{t-1}^B} -1) + 1 }{d(e^{{\alpha}_{t-1}^B} - 1) + 1} =1 $}.
        Hence, $ q=d $ is false.
        \\ \indent
        Second, assume that {\footnotesize $ q<d $},  i.e., {\footnotesize $ \sum{\bm{q}^+} < \sum{\bm{d}^+}$}.
        Since {\footnotesize $ \sum{(\bm{q}^+ + \bm{q}^-)} = \sum{(\bm{d}^+ + \bm{d}^-)} = 1$}, we have {\footnotesize $ \sum{\bm{q}^-} > \sum{\bm{d}^-}$}.
        That is, there exist {\footnotesize $ q_k \in  \bm{q}^-$}  and {\footnotesize $ d_k \in  \bm{d}^-$} satisfying {\footnotesize $ \frac{q_k}{d_k} >1 $}.
        This leads to a contradiction to Inequality \eqref{eq:cond2} because of  {\tiny $ \frac{q_k}{d_k} < \frac{q(e^{{\alpha}_{t-1}^B} -1) + 1 }{d(e^{{\alpha}_{t-1}^B} - 1) + 1} <1 $}.
        Hence, $ q<d $ is false.
        \\
        Therefore, {\footnotesize $ q>d $} is true.
        Then, we have {\tiny $ \frac{q(e^{{\alpha}_{t-1}^B} -1) + 1 }{d(e^{{\alpha}_{t-1}^B} - 1) + 1} >1 $}.
        By Inequality \eqref{eq:cond1}, it follows that  {\tiny $ \frac{q_j}{d_j} >\frac{q(e^{{\alpha}_{t-1}^B} -1) + 1 }{d(e^{{\alpha}_{t-1}^B} - 1) + 1} >1$} in which {\footnotesize $ q_j \in \bm{q}^+ $} and  {\footnotesize $ d_j \in \bm{d}^+ $}.
\end{proof}

\section{Proof of Corollary \ref{col:L_func}}
\label{appex:cor_L_func}
\begin{proof}
	The corollary is true because of Equation \eqref{eq:lfp_sol_b}.
\end{proof}

\section{Proof of Theorem \ref{thm:algo_extreme1}}
\label{appex:algo_extreme1}
\begin{proof}
	It is clear that, if $ q_i=d_i $ for $ i\in[1,n] $ in $ \bm{q} $ and $ \bm{d} $, Line 9 in Algorithm \ref{algo:cal_bpl} is always true so that the flag \textsf{update} is alway true until  all bits are removed from $ \bm{q}^+ $ and $ \bm{d}^+ $. 
	That is to say, $ \bm{q}^+ $ and $ \bm{d}^+ $ are empty.
	Hence, in this case, Algorithm \ref{algo:cal_bpl} will be terminated with empty $ \bm{q}^+ $ and $ \bm{d}^+ $, which result in $ q=d=0 $.
	Since for any two rows $ \bm{q} $ and $ \bm{d} $ in the transition matrix $ P_i $, we have $ q=d=0 $, it follows $ \mathcal{L}(\cdot)=0 $.
\end{proof}

\section{Proof of Theorem \ref{thm:algo_extreme2}}
\label{appex:algo_extreme2}
\begin{proof}
	For such two rows $  \bm{q }$ and $  \bm{d}$  in $ P_i $ that satisfy $ q_i=1 $ and $ d_i=0$, it is easy to see that Algorithm \ref{algo:cal_bpl} results in $ q=1 $, $ d=0 $ and the flag \textsf{update} is false.
	Hence, we have a maximum solution of $ \alpha $  (i.e., {\scriptsize ${ {{\alpha}^B_{t-1}}}$ or ${ {{\alpha}^F_{t+1}}}$}) for such LFP problem.
	Since the maximum solution of LFP problem w.r.t. any two other rows in $ P_i $ cannot be larger than $ \alpha $, $ \mathcal{L}(\cdot) $ is an identical function, i.e., $ \mathcal{L}(x)=x $.
\end{proof}

\section{Proof of Theorem \ref{thm:lfp_k_case}}
\label{appex:lfp_k_case}
\begin{proof}
	According to Corollary \ref{col:q_d}, we have {\footnotesize $ \frac{q(e^\alpha)+1}{d(e^\alpha)+1}  >1$} because of $ q>d $.
	By the pigeonhole principle, there only exists $ k $ cases as shown in the theorem.
	We only need to pove that the statements about $ q $ and $ d $ for each case are true.
	In case $ k $, it is clear that {\small $ \frac{q_1}{d_1} > \frac{q_1(e^\alpha)+1}{d_1(e^\alpha)+1} $}.
	In case $ k-1 $, the final $ \bm{q}^+ $ and $ \bm{d}^+ $ are impossible to include $ \{q_3, \cdots, q_n\} $ and $ \{d_3, \cdots, d_n\} $, respectively; otherwise, it violates Inequality \eqref{eq:cond1}. 
	Additionally, the final $ \bm{q}^+ $ and $ \bm{d}^+ $ in case $ k-1 $ are impossible to exclude any elements in $ \{q_1, q_2\} $ and $ \{d_1, d_2\} $, respectively; otherwise,  it violates Inequality \eqref{eq:cond2}. 
	Similarly, we can verfiy all other cases are true.
\end{proof}

\section{Proof of Theorem \ref{thm:lfp_seg_func}}
\label{appex:lfp_seg_func}
\begin{proof}
	We now show that the pairs of $ q $ and $ d $ will transit from the case $ 1 $ to case $ k $ in Theorem \ref{thm:lfp_k_case} when   $ \alpha $ is increasing.
	When $ \alpha =0$, it is easy to see that {\scriptsize $  \frac{q(e^0-1)+1}{d(e^0-1)+1} =1$}; hence,  $ q $ and $ d $ fall into  case $ 1 $.
	We can see that {\scriptsize $ \frac{q(e^\alpha-1)+1}{d(e^\alpha-1)+1}$} continues to increase along with $ \alpha $ untill  $ \alpha=\alpha_1 $ which makes {\scriptsize $ \frac{q(e^{\alpha_1}-1)+1}{d(e^{\alpha_1}-1)+1}=\frac{q_k}{d_k} $} (i.e., {\scriptsize $ \alpha_1 = \log (\frac{q_k-d_k}{q*d_k - d*q_k}+1)  $}), where {\scriptsize $q= q^k= \sum\nolimits_{i=1}^{k}q_i$} and {\scriptsize $d= d^k=\sum\nolimits_{i=1}^{k}d_i$}.
	If {\scriptsize $ \alpha = {\alpha_1} $},  $ q $ and $ d $ fall into case $ 2 $.
	Hence, when {\scriptsize $ 0 \leq \alpha <  \alpha_1$}, $ q $ and $ d $ are the ones in case $ 1 $.
	We can also find the \textit{transition point} $ \alpha_2 $ by solving {\scriptsize $ \frac{q(e^{\alpha_2}-1)+1}{d(e^{\alpha_2}-1)+1}=\frac{q_{k-1}}{d_{k-1}} $} (i.e., {\scriptsize $ \alpha_2 = \log (\frac{q_{k-1}-d_{k-1}}{q*d_{k-1} - d*q_{k-1}}+1)  $}), where {\scriptsize $q= \sum\nolimits_{i=1}^{k-1}q_i$} and {\scriptsize {\scriptsize $d= \sum\nolimits_{i=1}^{k-1}d_i$}}.
	If {\scriptsize $ \alpha = {\alpha_2} $},  $ q $ and $ d $ fall into case $ 3 $.
	Hence, when {\scriptsize $ \alpha_1 \leq \alpha <  \alpha_2$}, $ q $ and $ d $ are the ones in case $ 2 $.
	Similarly, we can use the same logic to find {\scriptsize $ \alpha_3, \cdots, \alpha_{k-1} $}, which are called \textit{transition points}.
	If {\scriptsize $ \alpha \geq \alpha_{k-1} $},  $ q $ and $ d $ fall into the case $ k $ because of {\scriptsize $ \frac{q_1}{d_1} > \frac{q_1(e^\alpha-1)+1}{d_1(e^\alpha-1)+1}  $} regardless of how large $ \alpha $ is.
	Therefore, given $ \bm{q} $ and $ \bm{d} $, {\scriptsize $ f_{\bm{q},\bm{d}}(\alpha)=   \frac{q(e^\alpha-1)+1}{d(e^\alpha-1)+1}$} is a piecewise function 
	whose sub-domains and coefficients of $ q $ and $ d $ are shown in Equation \eqref{eq:piecewise}.
\end{proof}

\section{Proof of Theorem \ref{thm:gen_L_func}}
\label{appex:gen_L_func}
\begin{proof}
	It is easy to see that functions $ f(\alpha) =\frac{q(e^\alpha -1)+1}{d(e^\alpha -1)+1}$ and $ f'(\alpha)=\frac{q'(e^\alpha -1)+1}{d'(e^\alpha -1)+1} $ are at most one intersection when $ \alpha>0 $.
	Assume there is an intersection in $ \left[a_1, a_2\right] $.
	Since {\small $ f(a_1) \geq f(a_1) $} and both of $ f$ and $ f'$ are increasing functions, we have {\small $ f(a_2) \leq f'(a_2) $}, which is in contradiction with the given condition {\small $ f(a_2) \geq f'(a_2) $}.
	Hence, the assumption fails and they have no intersection in $ \left[a_1, a_2\right] $.
	In other words, $ f(\alpha) $ is always larger than $ f'(\alpha) $  when $ a_1< \alpha<a_2 $.
\end{proof}

\section{Proof of Theorem \ref{thm:sup}}
\label{appex:thm_sup}
\begin{proof}
        We denote the supremum of BPL or FPL over time by $ \alpha $.
        If $ d=0 $ and $ q=1 $, it is easy to see that $ \mathcal{L}^B(\cdot) = \mathcal{L}^F(\cdot) =1 * \cdot  $.
        In other words, BPL and FPL will increase linearly without limit; hence, $ \mathcal{F} $ does not exist.
        If BPL (or FPL) has a limit, since the current privacy leakage should be calculated based on the previous one (or the next one for FPL), we have $ \alpha = \mathcal{L}^B(\alpha) +\epsilon $ (see Equations \eqref{eq:bpl_f} or \eqref{eq:fpl_f}), namely {\footnotesize $ \alpha = \log \frac{q(e^\alpha-1)+1}{d(e^\alpha-1)+1} + \epsilon$}. 
        By expanding this equation and letting $ \mathcal{F} $ be $ e^\alpha $. we have an equation with a variable $ \mathcal{F} $, which $ \mathcal{F} >1$.
        \begin{equation}
        d\mathcal{F}^2+(qe^\epsilon + d -1)\mathcal{F} + (qe^\epsilon  - e^\epsilon ) = 0 \label{eq:thm_sup_proof} 
        \end{equation}
        
        Hence, the logarithm of the solution (if it exists)  in the above equation is the supremum of BPL or FPL.
        If $ d=0 $ and $ q\neq1 $, we have $\mathcal{F} = \frac{(1-q) e^\epsilon}{1-q e^\epsilon } $.
        To ensure a positive $ \mathcal{F} $, because of $ (1-q) e^\epsilon > 0 $, we need $ 1-q e^\epsilon >0 $, i.e., $ \epsilon <\log (1/q) $.
        Therefore, when $ d=0 $, $ q\neq 1 $ and $ \epsilon \geq \log (1/q) $, the supremum does not exist;
        when $ d=0 $, $ q\neq 1 $ and $ \epsilon <\log (1/q) $, we have the supremum $\log \frac{(1-q) e^\epsilon}{1-q e^\epsilon }  $.
        If $ d \neq 0 $, by solving the quadratic equation \eqref{eq:thm_sup_proof}, we can obtain a positive solution:
        $  {\tiny \frac{{\sqrt { 4d{e^\epsilon }(1-q ) + {{( d + q{e^\epsilon } -1)}^2} }  + d + q{e^\epsilon } - 1}}{{2d}} } $.
        It is easy to verify that  this solution always exist and is greater than $ 1 $.
        Therefore, the theorem follows.
\end{proof}

\end{document}